\documentclass[sigconf]{acmart}
\usepackage{enumitem}
\usepackage{amsmath}
  
\usepackage{tcolorbox}
\usepackage{amssymb}
\usepackage{bm}
\usepackage{nicefrac}
\usepackage{booktabs}
\usepackage{array}
\usepackage{multirow}
\usepackage{threeparttable}
\usepackage{makecell}
\usepackage[procnumbered,ruled,vlined,linesnumbered]{algorithm2e}
\usepackage{siunitx}
\usepackage{stfloats}
\usepackage{graphicx}
\usepackage{subfigure}
\usepackage{hyperref}
\usepackage{enumerate}
\usepackage{graphicx}
\usepackage{caption}
\usepackage{subcaption}

\newtheorem{theorem}{Theorem}[section]

\newtheorem{lemma}[theorem]{Lemma}

\newenvironment{fminipage}%
{\begin{Sbox}\begin{minipage}}%
		{\end{minipage}\end{Sbox}\fbox{\TheSbox}}

\def\calG{\mathcal{G}}

\def\calF{\mathcal{F}}
\def\calR{\mathcal{R}}
\newcommand{\removelatexerror}{\let\@latex@error\@gobble}

\newcommand\LL{\bm{\mathit{L}}}

\newcommand\zz{\boldsymbol{\mathit{z}}}

\newcommand\oomega{\boldsymbol{\mathit{\omega}}}

\newcommand\ee{\boldsymbol{\mathit{e}}}

\newcommand\sss{\boldsymbol{\mathit{s}}}

\renewcommand\AA{\boldsymbol{\mathit{A}}}

\newcommand\DD{\boldsymbol{\mathit{D}}}

\newcommand\II{\boldsymbol{\mathit{I}}}

\DontPrintSemicolon
\SetKw{KwAnd}{and}
\SetFuncSty{textsc}
\SetKwInOut{Input}{Input\ \ \ \ }
\SetKwInOut{Output}{Output}

\usepackage{tabularx}
\usepackage{stfloats}

\DontPrintSemicolon
\SetKw{KwAnd}{and}
\SetFuncSty{textsc}
\SetKwInOut{Input}{Input\ \ \ \ }
\SetKwInOut{Output}{Output}

\usepackage{tabularx}
\usepackage{stfloats}

\copyrightyear{2024}
\acmYear{2024}
\setcopyright{acmlicensed}\acmConference[WWW '24]{Proceedings of the ACM Web Conference 2024}{May 13--17, 2024}{Singapore, Singapore}
\acmBooktitle{Proceedings of the ACM Web Conference 2024 (WWW '24), May 13--17, 2024, Singapore, Singapore}
\acmDOI{10.1145/3589334.3645578}
\acmISBN{979-8-4007-0171-9/24/05}

\begin{document}
%	\fancyhead{}
	\title{Efficient Computation for Diagonal  of  Forest Matrix via Variance-Reduced Forest Sampling}%Maximizing the Social Opinions by Creating Relationships
%\footnotemark[1]
\iffalse
\author{Haoxin Sun and Zhongzhi Zhang\footnotemark}

\email{21210240097@m.fudan.edu.cn, 
 zhangzz@fudan.edu.cn}
\fi

	\author{Haoxin Sun}
	\affiliation{%
		\institution{Fudan University}
		%\streetaddress{1 Th{\o}rv{\"a}ld Circle}
		\city{Shanghai}
		\country{China}}
	\email{21210240097@m.fudan.edu.cn}
	
	\author{Zhongzhi Zhang\footnotemark}
	\affiliation{%
		\institution{Fudan University}
		%\streetaddress{1 Th{\o}rv{\"a}ld Circle}
		\city{Shanghai}
		\country{China}}
	\email{zhangzz@fudan.edu.cn}

\begin{abstract}

The forest matrix of a graph, particularly its diagonal elements, has far-reaching implications in network science and machine learning. The state-of-the-art algorithms for the diagonal of forest matrix computation are based on the fast Laplacian solver. However, these algorithms encounter limitations when applied to digraphs due to the incapacity of the Laplacian solver. To overcome the issue, in this paper, we propose three novel sampling-based algorithms: \textsc{SCF}, \textsc{SCFV}, and \textsc{SCFV+}. Our first algorithm \textsc{SCF}  leverages a probability interpretation of the diagonal of the forest matrix and utilizes an extension of Wilson's algorithm to sample spanning converging forests. To reduce the variance in the forest sampling, we develop two novel variance-reduced techniques. The first technique, leading to the proposal of the \textsc{SCFV} algorithm, is inspired by opinion dynamics in graphs and applies matrix-vector iteration to the spanning forest sampling. While \textsc{SCFV} achieves reduced variance compared to \textsc{SCF}, the cross-product term in its variance expression can be complex and potentially large in certain graphs. Therefore, we develop another technique, leading to a new iteration equation and the \textsc{SCFV+} algorithm. \textsc{SCFV+} achieves further reduced variance without the cross-product term in the variance of \textsc{SCFV}. We prove that \textsc{SCFV+} can achieve a relative error guarantee with high probability and maintain a linear time complexity relative to the number of nodes in the graph, presenting a superior theoretical result compared to state-of-the-art algorithms. Finally, we conduct extensive experiments on various real-world networks, showing that our algorithms achieve better estimation accuracy and are more time-efficient than the state-of-the-art algorithms. Particularly, our algorithms are scalable to massive graphs with more than twenty million nodes in both undirected and directed graphs.

\end{abstract}



	%% A "teaser" image aNeNepears between the author and affiliation
%	\begin{teaserfigure}
%		\includegraphics[width=\textwidth]{sampleteaser}
%		\caption{Seattle Mariners at Spring Training, 2010.}
%		\Description{Enjoying the baseball game from the third-base
%			seats. Ichiro Suzuki preparing to bat.}
%		\label{fig:teaser}
%	\end{teaserfigure}
	

\maketitle
%\renewcommand{\thefootnote}{\fnsymbol{footnote}}
% \renewcommand{\thefootnote}{*}

\renewcommand{\thefootnote}{*}
\footnotetext[1]{Corresponding author. Both authors are with Shanghai Key Laboratory of Intelligent Information Processing, School of Computer Science, Fudan University, Shanghai 200433, China. Zhongzhi zhang is also with Institute of Intelligent Complex Systems, Fudan University, Shanghai 200433,China.}

\section{Introduction}

 As a typical representation of a graph, the Laplacian matrix  $\LL$ encapsulates much useful structural and dynamical information of the graph~\cite{Me94}. In addition to $\LL$ itself, the forest matrix, denoted as $\mathbf{\Omega} = (\II+\LL)^{-1}$, is also a powerful tool in network science, with close ties to spanning rooted forests in graphs~\cite{ChSh06}. Recent studies have spotlighted various applications of the forest matrix and its variants in various aspects, such as  Markov processes~\cite{AvLuGaAl18,AvCaGaMe18}, opinion dynamics~\cite{GiTeTs13, XuBaZh21,SuZh23}, and graph signal processing~\cite{PiAmBaTr21,PiAmBaTr20}. Particularly, the diagonal entries of $\mathbf{\Omega}$ appear frequently in diverse applications. First, it can serve as the forest closeness centrality~\cite{JiBaZh19, GrAnPrMe21} of a network. Besides,  the diagonal of the forest matrix has been closely associated with determinantal point processes in machine learning~\cite{KuTa12}, and  has found relevance through electrical interpretations in multi-agent and network-based problems~\cite{RoFrFa17}.

% As a typical representation of a graph, the Laplacian matrix  $\LL$ encapsulates much useful structural and dynamical information of the graph~\cite{Me94}. 

In order to achieve better effects of the applications for the diagonal entries of $\mathbf{\Omega}$ for a graph with $n$ nodes, the first step is to compute or evaluate the diagonal of $\mathbf{\Omega}$. A straightforward computation of $\mathbf{\Omega}$ involves inverting  matrix $\II+\LL$, which costs $O(n^3)$ operations and $O(n^2)$ memories and thus is prohibitive for relatively large graphs. In previous work, two Laplacian solver~\cite{CoKyMiPaJaPeRaXu14} based algorithms, \textsc{JLT} and \textsc{UST}, were proposed to compute the diagonal of $\mathbf{\Omega}$~\cite{JiBaZh19, GrAnPrMe21}. Although these methods outperform the standard approach, they are constrained by the Laplacian solver's inability to handle directed graphs. Moreover, while \textsc{UST} exhibits superior performance compared to \textsc{JLT} in experiments~\cite{GrAnPrMe21}, it provides only an absolute error guarantee theoretical analysis, whereas  \textsc{JLT} offers a relative error guarantee~\cite{JiBaZh19}. Consequently, a theoretically guaranteed estimation algorithm for approximating diagonal of $\mathbf{\Omega}$ for both undirected and directed graphs is imperative.

In this paper, we delve deep into the problem of efficiently computing the diagonal of the forest matrix in a digraph  with $n$ nodes, in order to overcome the challenges and limitations of existing algorithms. The main contributions of this work are summarized as follows:

(i) We introduce two forest interpretations of the diagonal of the forest matrix, from the perspectives of average tree size and rooted probability, respectively. Based on the probability interpretation, we develop \textsc{SCF}, an algorithm that leverages an extension of Wilson's algorithm to approximate the diagonal of the forest matrix, offering a time complexity of $O(ln)$, where $l$ represents the sampling number.

(ii) To reduce the variance in forest sampling, we develop two novel variance-reduced techniques and  propose two algorithms \textsc{SCFV} and \textsc{SCFV+}. We prove that the variance of the estimators in the three algorithms diminishes progressively. Notably, \textsc{SCFV+} not only achieves a relative error guarantee with high probability but also maintains a linear time complexity relative to the number of nodes in the graph, thereby presenting a superior theoretical result compared to existing algorithms.

(iii) Extensive experiments on various undirected and directed  networks demonstrate that  our algorithms achieve better estimation accuracy and enhance computation efficiency, compared to the state-of-the-art algorithms. Moreover, our algorithms are scalable to massive graphs with over thirty million nodes.

\section{Related Work}
% In this section, we briefly review the existing work related to ours.

Identifying crucial nodes in a graph is a fundamental issue with a rich history in machine learning and graph analysis~\cite{LiPeShYiZh19,SkRaMiYo19,GrGrLa18,SkMiRa18,Ne10}. 
  Various metrics and indices have been developed to quantify the relative importance or centrality of nodes within a network~\cite{BaZh22},  including but not limited to,  betweenness centrality~\cite{Fr77}, closeness centrality~\cite{Ba48,Ba50}, eccentricity~\cite{ChSchWe96}. In addition to the classic centrality measures, forest closeness centrality has been proposed and explored for its unique advantages~\cite{JiBaZh19}. One of the notable advantages of forest closeness centrality is its applicability to disconnected networks, which is particularly relevant for various real-world networks such as Mobile Ad hoc Networks~\cite{DaHa07} and protein-protein interaction networks~\cite{HuLiWu18}. Furthermore,  the forest centrality has a better discriminating power than alternate metrics such as betweenness, harmonic centrality, eigenvector centrality, and PageRank~\cite{BaZh22}. 

The calculation of forest closeness centrality is inherently tied to the diagonal elements of the forest matrix  $\mathbf{\Omega}$. To speed up the computation, a nearly linear time algorithm \textsc{JLT} was proposed in~\cite{JiBaZh19}, which combines the Johnson-Lindenstrauss lemma \cite{JoLi84,Ac03} with the  fast Laplacian solver~\cite{CoKyMiPaJaPeRaXu14}, necessitating a time complexity of $O(m\epsilon^{-2}\log^{2.5}n\log\frac{1}{\epsilon}{\rm polyloglog}(n))$ to achieve a relative error bound.  Subsequently, the \textsc{UST} algorithm was proposed in~\cite{GrAnPrMe21}, combining a single instance of the Laplacian solver and uniform spanning tree sampling. Compared to \textsc{JLT}, \textsc{UST} achieves computational acceleration  in experiments and needs a total time complexity of \(\widetilde{O}(m\epsilon^{-2}\log^{3/2}n)\) to guarantee an absolute error of \(\epsilon\) with high probability. However, both algorithms utilize the fast Laplacian solver~\cite{CoKyMiPaJaPeRaXu14}, which is not applicable to digraphs.

% The forest matrix is closed related to the spanning rooted forest in graphs~\cite{ChSh06,ChSh97}.
Wilson's algorithm plays a pivotal role in our three algorithms \textsc{SCF}, \textsc{SCFV}, and \textsc{SCFV+}. Initially proposed to sample a spanning tree in graphs~\cite{Wi96}, Wilson's algorithm and its variants have found applications across various domains, such as computing the PageRank vector~\cite{LiLiDaChQiWa23PageRank,LiLiDaWa22}, solving linear systems in graph signal processing~\cite{PiAmBaTr21,PiAmBaTr20}, addressing optimization problems in opinion dynamics~\cite{SuZh23}, and estimating effective resistance~\cite{GrAnPrMe21,LiLiDaChQiWa23Resistance}. While several variance reduction techniques have been proposed and utilized in various sampling-based problems~\cite{LiLiDaChQiWa23PageRank,PiAmBaTr22,PiAmBaTr22trace}, these techniques are either unsuitable for our diagonal estimation problem or induce a prohibitively high complexity. Thus, developing  novel estimators with reduced variance for the diagonal of the forest matrix, which are applicable to both undirected graphs and digraphs, is of significant importance and  becomes the primary research subject of this paper.

\section{Preliminaries}

 % In this section, we introduce some useful notations and tools for the convenience of description and analysis.
% \subsection{Notations}

% We use normal lowercase letters like $ a,b,c $ to denote scalars in set of real numbers, normal uppercase letters like $ A,B,C $ to denote sets, bold lowercase letters like $ \aaa,  \bb, \cc$ to denote column vectors, and bold uppercase letters like $ \AA,\BB,\CC $ to denote matrices.  We use $\AA^{\top}$ and  $\aaa^{\top}$ to represent the transpose of $\AA$ and  $\aaa$, respectively. Let $ \ee_i $  denote the column vector of appropriate dimension, where the $ i $-th element is $ 1 $, and other elements are $ 0 $. Let $ \mathbf{0} $ be an appropriate-dimension column vector with all entries being zeros, and let $\mathbf{1}$ be an appropriate-dimension column vector with all entries being ones. Let $ \II $ denote an appropriate-dimension identity matrix.  For a matrix $\AA$, $\AA_{H,F} $ denotes the submatrix of $\AA $ with row indices in set $H$ and column indices in set $ F$, and $ A_{-H} $ denotes the submatrix of $ \AA $ obtained from $ \AA $ by deleting rows and columns corresponding to nodes in set $ H $. Similarly, for a vector $ \aaa $, we use $ \aaa_{-H} $ to denote the vector obtained from $ \aaa $ by deleting elements in set $H$. If $ H $ contains only a single element $i$, we use $ \AA_{-i}$ and $\aaa_{-i}$ to denote, respectively, $ \AA_{-\{i\}} $ and $ \aaa_{-\{i\}} $ for simplicity.

\subsection{ Graph and  Laplacian Matrix}
Consider  an unweighted simple directed graph (digraph) $\calG=(V,E)$  with $n=|V|$ nodes (vertices) and $m=|E|$ directed edges (arcs), where \(V=\{v_1,v_2,\ldots,v_n\}\) represents the node set and  \(E=\{(v_i, v_j)\mid v_i, v_j \in V\}\) signifies the set of directed edges. An arc \( (v_i,v_j) \in E \) indicates a directed edge pointing from node \(v_i\) to node \(v_j\).  In what follows, $v_i$ and $i$ are used interchangeably to represent node $v_i$ if incurring no confusion. Let $N(i) $ be the node set accessible from node $ i $. In other words, $N(i) =\{ j: (i,j)\in E\}$.  A digraph is called weakly connected if it is connected when one replaces any directed edge $(i,j)$ with two directed edges $(i,j)$ and $(j,i)$ in opposite directions. 
%, and \(E\) signifies the set of directed edges such that \(E=\{(v_i, v_j)\mid v_i, v_j \in V\}\). 

The structure information of digraph $\calG=(V,E)$  is characterized by  its adjacency matrix $\AA=(a_{ij})_{n\times n}$, where $a_{ij} = 1$  if $ (v_i,v_j) \in E $ and $a_{ij} = 0 $ otherwise. For any node \(i\) in \(\calG\), its in-degree \(d^+_i\) and out-degree \(d^-_i\) are given by \(d^+_i=\sum_{j=1}^n a_{ji}\) and \(d^-_i=\sum_{j=1}^n a_{ij}\), respectively. In the sequel, we use $d_i$ to represent the out-degree $d_i^-$. The diagonal out-degree matrix of digraph $\calG$ is defined as ${\DD} = {\rm diag}(d_1, d_2, \ldots, d_n)$, and the Laplacian matrix of digraph $\calG$ is defined to be ${\LL}={\DD}-{\AA}$. Let $\II$ be the $n$-dimensional identity matrix, and $\ee_i$ be the  $i$-th standard basis column vector, with $i$-th element being 1 and other elements being 0.

% A path $P$ from node $v_1 $ to $ v_k $ is an alternating sequence of nodes and arcs $v_1$,$(v_1,v_2)$,$v_2$,$\cdots$, $v_{j-1},(v_{j-1}$,$v_j)$, $v_j$ in which nodes are distinct and every arc $ (v_i,v_{i+1}) $ is from $ v_i $ to $ v_{i+1}$. A loop  is a path plus an arc from the ending node to the starting node. A digraph is (strongly) connected if for any pair nodes $v_x$ and $v_y$, there is a path from $v_x$ to $v_y$, and there is a path from $v_y$ to $v_x$ at the same time. A digraph is called weakly connected if it is connected when one replaces any directed edge $(i,j)$ with two directed edges $(i,j)$ and $(j,i)$ in opposite directions. A tree is a weakly connected graph with no loops. An isolated node is considered as a tree. A forest is a particular graph that is a disjoint union of trees.

% In a digraph $\calG$, if for any arc $(i,j)$, the arc $(j,i)$ exists, $\calG$ is reduced to an undirected graph. When $\calG $ is undirected, $a_{ij}= a_{ji}$ holds for an arbitrary pair of nodes $i$ and $j$, and thus $d^+_i= d^-_i$ holds for any  node $i\in V$. Moreover, in undirected graph $\calG $ both adjacency matrix $\AA$ and Laplacian matrix $\LL$ of  $\calG $ are symmetric, satisfying  $\LL\mathbf{1}=\mathbf{0}$.
%\subsection{Spanning Converging Forests and Forest Matrix}

\subsection{Spanning Forests and Forest Matrix}

For a digraph $\calG=(V,E)$, a spanning subgraph of $\calG$  retains all nodes from \(V\) but may only include a subset of edges from \(E\). A rooted converging tree is a weakly connected digraph, where one node, called the root node, has an out-degree of 0, and all other nodes have an out-degree of 1. An isolated node is considered as a converging tree with the root being itself. A spanning converging forest of digraph $\calG$ is a spanning subgraph of $\calG$, whose weakly connected components are rooted converging trees. Such a forest aligns with the concept of an in-forest in~\cite{AgCh01,ChAg02}.

The forest matrix~\cite{ChSh97,ChSh98} is defined as $\mathbf{\Omega}=\left(\II+\LL\right)^{-1}=(\omega_{ij})_{n \times n}$. In the context of digraphs, the forest matrix $\Omega$ is row stochastic, with all its components in the interval $[0,1]$. Moreover,  for each column, the diagonal elements surpass the other elements, that is  $ 0\leq\omega_{ji}< \omega_{ii}\leq 1$ for any pair of nodes $ i$ and $j $, and the diagonal element $\omega_{ii}$ of matrix $\mathbf{\Omega}$ satisfies $\frac{1}{1+d_i}\leq \omega_{ii} \leq \frac{2}{2+d_i}$~\cite{SuZh23}. In the sequel,  we use  $\oomega$ to denote the column vector of all diagonal elements of the forest matrix, that is $\oomega = (\omega_{11},\cdots,\omega_{nn})^{\top}$.

\section{Fast Forest Sampling Algorithm}

 In this section, we introduce two interpretations of the forest matrix diagonal. Utilizing the probability interpretation, along with an extension of Wilson's algorithm, we propose a fast sampling algorithm to calculate the diagonal element vector $\oomega$ of the forest matrix.

% \section{Properties of Forest Matrix}
% In this section, we provide some properties of the forest matrix $\Omega$, and give a new interpretation of the reciprocal of the diagonal elements of the forest matrix from the perspective of the number of nodes in each rooted spanning tree. 

\subsection{Novel Forest Interpretation for Diagonal  of Forest Matrix}
In this subsection, we introduce a novel interpretation for the diagonal of forest matrix $\mathbf{\Omega}$. Before proceeding, we introduce some essential notations.

For an unweighted digraph $\calG=(V,E)$, let $\calF $ denote the set of all spanning converging forests. For a given spanning converging forest $\phi\in\calF$,  define the root set $\mathcal{R}(\phi )$ of $\phi$ as the collection of roots from all converging trees that constitute $ \phi$, that is, $\mathcal{R}(\phi ) = \{i:(i,j) \notin \phi, \forall j \in V_{\phi} \}$. Since each node $i$ in $\phi$  belongs to a specific converging tree, we define a function $r_{\phi}(i): V \rightarrow \calR(\phi) $ mapping node $i$ to the root of its associated converging tree. Thus, if $r_{\phi}(i) = j$, it implies that $j$ is in $\calR(\phi)$, and both nodes $i$ and $j$ are in the same converging tree  in \(\phi\). Define $ \calF_{ij} $ as the set of spanning converging forests in which nodes $i$ and $j$ are within the same converging tree, rooted at node $j$. Formally, $\calF_{ij} = \{\phi: r_{\phi}(i) = j, \phi \in \calF\}$.  It follows that $\calF_{ii} = \{\phi: i\in \calR(\phi), \phi \in \calF\}$. For example,  the left-hand side of Figure~\ref{ftoy} is a toy digraph consisting of  5 nodes and 8 edges, while the right-hand side illustrates one of its spanning converging forests with roots marked in blue. Let $\phi$ represent the spanning converging forest shown in Figure~\ref{ftoy}. With the above notations, we have $\calR(\phi) = \{3,5\}$, and $r_{\phi}(1) = 3 $.

\begin{figure}[htbp!]
	\centering
	\includegraphics[width=.5\columnwidth]{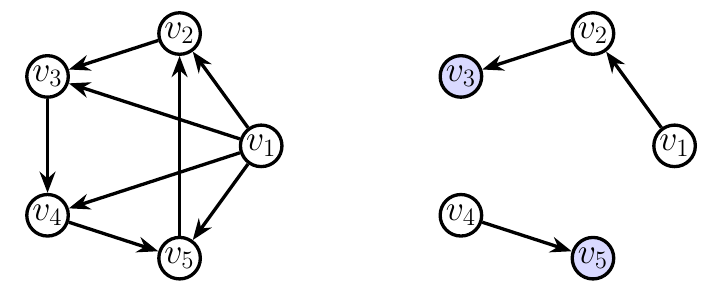}
	\caption{A toy digraph  and one of its spanning converging forest $\phi$.}\label{ftoy}	
\end{figure}

For a node $i \in V$ and a  spanning converging forest $\phi \in \calF_{ii} $, let $N(\phi,i) $ be a set defined by $N(\phi,i) = \{ j: r_\phi(j) = i \}$. By definition, for any $\phi \in \calF_{ii}  $, $ |N(\phi,i)| $ is equal to the number of nodes in the converging tree in $\phi$, whose root is node $i$. For two nodes $i$ and $j$ and a  spanning converging forest $\phi$, define $\mathbb{I}_{\{r_{\phi}(i)=j\}}$ as an indicator function, which is 1 if the input statement is true and 0 otherwise. For example, if $r_{\phi}(i)=j$, $\mathbb{I}_{\{r_{\phi}(i)=j\}}=1$, and $ \mathbb{I}_{\{r_{\phi}(i)=j\}}=0$ otherwise.

With the above defined notations, we present an interpretation of the diagonal of the forest matrix $\mathbf{\Omega}$ in undirected graphs.  The following theorem establishes the relationship between the reciprocal of $\omega_{ii}$ and the average number of nodes rooted at node $i$.

\begin{theorem}\label{th-reciprocal}
    For an undirected graph $G=(V,E)$, the reciprocal of the $i$-th diagonal elements of forest matrix $\mathbf{\Omega}$ is equal to the average size of the tree containing node $i$ across all converging spanning forest where $i$ serves as one of the root nodes. Formally, this can be expressed as $\frac{1}{\omega_{ii}} = \frac{\sum_{\phi\in\calF_{ii}}N(\phi,i)}{|\calF_{ii}|}.$
\end{theorem}

% In the following subsection, the discussion is focused on the scenario where $G$ is undirected. Specifically, for any nodes $i,j\in V$, if $(i,j)\in E$, it follows that $(j,i)\in E$. Consequently, the forest matrix $\Omega$ is symmetric and double stochastic. It is important to note that the notations defined in the preliminary section remain applicable, as undirected graphs are merely a specific subset of directed graphs.

It has been shown~\cite{JiBaZh19} that  the $i$-th diagonal element $\omega_{ii}$ of the forest matrix $\mathbf{\Omega}$ is consistent with the forest distance from node $i$ to other nodes, with a smaller sum of distances indicating a more pivotal node. In Theorem~\ref{th-reciprocal}, we introduce a novel interpretation for $\omega_{ii}$, suggesting that if the average tree size rooted at $i$ is larger, then $\frac{1}{\omega_{ii}}$ will be larger, $\omega_{ii}$ will be smaller, and consequently, node $i$ will be more significant. This aligns with the analysis in~\cite{JiBaZh19}. However, both the distance interpretation and the average tree size interpretation are valid only for undirected graphs, as their derivations utilize the symmetry of the forest matrix, a property exclusive to undirected graphs. Below we propose another interpretation from a probabilistic perspective, which is applicable to both undirected graphs and digraphs and inspires us  for the design of our sampling algorithms.

\subsection{ Probability Interpretation and Extension of Wilson's Algorithm}

The entries of the forest matrix are closely related to the spanning converging forest in graphs. Using the approach in~\cite{Ch82,ChSh06,ChSh97},  the entry $\omega_{ij}$ of the forest matrix $\mathbf{\Omega}$ can be expressed as $\omega_{ij}= |\calF_{ij}|/|\calF|$.  By setting $i=j$, we have $\omega_{ii} = {|\calF_{ii}|}/{|\calF|}$ for every node $i\in V$, which leads to  a probabilistic interpretation of the diagonal entry $\omega_{ii}$ of the forest matrix. Specifically,  $\omega_{ii}$ represents the probability that node $i$ is included in the root set $\calR(\phi)$, when a spanning converging forest $\phi\in\calF$ is sampled uniformly.  Then for a spanning converging forest $\phi\in\calF$, we can define an estimator $\widehat{\omega}_{ii}(\phi)$ of $\omega_{ii}$ as $\widehat{\omega}_{ii}(\phi)  =\mathbb{I}_{\{i\in\calR(\phi)\}}  $. The estimator $\widehat{\omega}_{ii}$ is unbiased since $\mathbb{E}\{\widehat{\omega}_{ii}(\phi) \}= \mathbb{P}\{i\in\calR(\phi) \} = \frac{|\calF_{ii}|}{|\calF|} = \omega_{ii}.$

Therefore, if we  uniformly generate a spanning converging forest in $\calG$, and record the probability of $i$ serving as a root node, we can estimate $\omega_{ii}$.  In the following, we give a brief introduction of an extension of Wilson's Algorithm in order to uniformly sample spanning converging forest $\phi \in \calF$.
\begin{figure}[hbp!]
	\centering
	\includegraphics[width=.6\columnwidth]{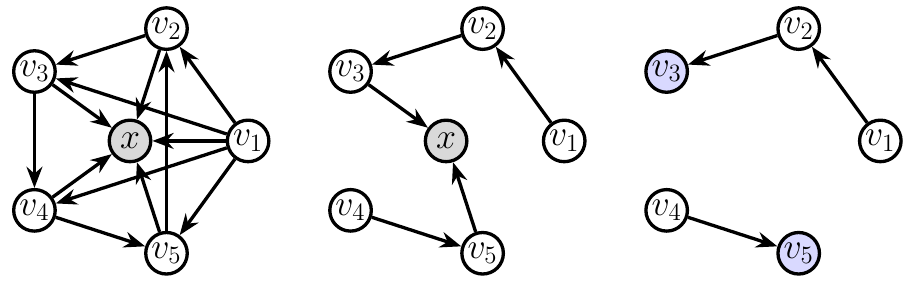}
	\caption{Illustration for a spanning converging forest generated using the extension of Wilson's algorithm for the toy graph in Figure~\ref{ftoy}.}\label{ftoyrf}	
\end{figure}

Wilson proposed an algorithm  based on a loop-erased random walk to get a  spanning tree rooted at a given node~\cite{Wi96}.  The loop-erasure technique, pivotal to this algorithm,  is a process derived from the random walk  by performing an erasure operation on its loops in chronological order~\cite{LaFr79,La80}. For a digraph $\calG= (V,E)$,  we can also apply an extension of Wilson's Algorithm to get a spanning converging forest $\phi \in \calF$, by using the method similar to that in~\cite{AvLuGaAl18,PiAmBaTr21,SuZh23}, which includes the following  steps. Firstly, construct an augmented digraph $\calG' $ by adding one new node $x$. Then for each node $i$ in the original graph, add new edge $(i,x)$ to the augmented graph $\calG' $. Subsequently, use Wilson's algorithm to generate a rooted spanning tree  in the augmented graph $\calG'$, designating $x$ as the root node. Finally, delete node $x$  and all edges connected to it in the rooted spanning tree, and define the root set $\calR$ as the nodes with an out-degree of 0, thereby obtaining a spanning converging forest in $\calG$. For example, in Figure~\ref{ftoyrf}, the left-hand side is the augmented toy graph, with one new node $x$ and some new edges added. The middle of Figure~\ref{ftoyrf} illustrates a rooted spanning tree in the augmented graph, rooted at $x$. The right-hand side is a spanning converging forest in the original toy graph,  obtained by deleting node $x$ and its connected edges.

Since Wilson's algorithm returns a uniform rooted spanning tree~\cite{Wi96}, the spanning converging forest obtained using the above steps is also uniformly selected from $\calF$.

\subsection{Fast Sampling Algorithm}
In this subsection, we propose a fast sampling algorithm to estimate the diagonal of the forest matrix based on the above-mentioned extended Wilson's algorithm. Additionally, we provide a theoretical analysis concerning its time complexity and relative error guarantee.

In the above, we have defined an unbiased estimator $\widehat{\omega}_{ii}(\phi)$, and introduced how to employ the extension of Wilson's algorithm to uniformly generate a spanning converging forest  $\phi$ for $\calG$.  Then  we can generate $l$ spanning forest $\phi_{1},\cdots,\phi_l$, and use the average value $\frac{1}{l}\sum_{j=1}^l\widehat{\omega}_{ii}(\phi_j)$ to approximate $\omega_{ii}$. We detail this in  Algorithm~\ref{alg-rf}.

\begin{algorithm}[htbp!]
	\caption{$\textsc{SCF}(\calG,l)$}
	\label{alg-rf}
	\Input{ $\calG$:a digraph  \\ $l$:number of generated spanning converging forest
	}
	\Output{$\widehat{\boldsymbol{\omega}}$ : a vector approximating the diagonal elements of the forest matrix}
 \textbf{Initialize} :
	$\widehat{\boldsymbol{\omega}}$[$i$] $ \leftarrow 0 $, $i= 1,2,\ldots,n $\\
 \For{$t = 1$ to $l$}{
 $\phi_t$ $\leftarrow$ a spanning converging forest generated from $\calG$\\
 \For{$ i = 1 $ to $ n $}{
 $j \leftarrow r_{\phi_t}(i)$\\
                 \If{$i = j$}{
                $\widehat{\boldsymbol{\omega}}$[$i$]$\leftarrow$ $\widehat{\boldsymbol{\omega}}$[$i$] $ + \frac{1}{l}$\\

                }
 }
% InForest[$ i $] $ \leftarrow $ false , $i= 1,2,\ldots,n $\;
% 	Next[$ i $] $ \leftarrow -1$ , $i= 1,2,\ldots,n $\;
%         RootIndex[$i$] $ \leftarrow 0 $, $i= 1,2,\ldots,n $\\
% 	\For{$ i = 1 $ to $ n $}
% 	{$ u \leftarrow i $\; 
% 		\While{not InForest[$ u $]}{
% 			seed $ \leftarrow $ \textsc{Rand}()  \; 
% 			\If{seed  $\leq  \frac{1}{1+d_u} $}{
% 				InForest[$ u $] $ \leftarrow $ true\;
% 				Next[$ u $]$ \leftarrow -1 $\;
%     RootIndex[$ u $]$ \leftarrow u $\;
% }
% 			\Else{
% 				Next[$ u $] $ \leftarrow $ \textsc{RandomSuccessor}($ u,\calG $)\; 
% 				u $ \leftarrow $ Next[$ u $]\;
% 			}
% 		}
%             RootNow $\leftarrow  $ RootIndex[$u$]\;
%             $ u \leftarrow i $\\
% 		\While{not InForest[$ u $]}{
% 			InForest[$ u $] $ \leftarrow $ true\;
%                 \If{$u \in N(\rm{Rootnow})$}{
%                 $\widebar{\boldsymbol{\omega}}$[$u$]$\leftarrow$ $\widebar{\boldsymbol{\omega}}$[$u$] $ + \frac{1}{l(1+d_u)}$\\

%                 }
                
% 			u $ \leftarrow $ Next[$ u $]\;
% 		}
% 	}
	 }
\textbf{return} $\widehat{\boldsymbol{\omega}}$ \;
\end{algorithm}

We now analyze the time complexity of Algorithm~\ref{alg-rf}.

% Before doing this, we present some properties of the diagonal element $\omega_{ii}$ of matrix $\mathbf{\Omega}$ for all nodes $i \in V$.
% \begin{lemma}\cite{SuZh23}\label{le-omega}
% 	For any $i=1,2,\ldots, n$, the diagonal element $\omega_{ii}$ of matrix $\mathbf{\Omega}$ sastisfies $\frac{1}{1+d_i}\leq \omega_{ii} \leq \frac{2}{2+d_i}$.
% \end{lemma}	

 % \begin{proof}
	% Let $l_{ij} $ be the entry at row $ i $ and column $ j $ of the Laplacian matrix $\LL $. Then $l_{ii} = d_i $, and $ l_{ij} = 0 $ or $ -1 $ for $i \neq j$. Using $\mathbf{\Omega}\left(\II+\LL\right)  = \II$,  one obtains that for any $i=1,2,\ldots, n$,
	% \begin{equation*}\label{key}
	% 1 = (1+d_i)\omega_{ii}+\sum_{u \neq i} \omega_{iu}l_{ui}\leq (1+d_i)\omega_{ii},
	% \end{equation*}
	% which implies that $ \omega_{ii}\geq \frac{1}{1+d_i} $ with  the equality holding  under certain conditions, for example,    $(u,i) \notin E $ for arbitrary $ u \in V $.
	% On the other hand, considering the fact that $ \sum_{j = 1}^n \omega_{ij} =1 $, one has 
	% \begin{equation*}
	% 1 =(1+d_i)\omega_{ii}+\sum_{u \neq i} \omega_{iu}l_{ui}\geq (1+d_i)w_{ii}-(1-w_{ii}),
	% \end{equation*}
	% which means  $\omega_{ii} \leq \frac{2}{2+d_i}$. The equality holds  under some conditions, for example,  $ (u,i) \in E $ for any $u \in V$.
 % \end{proof}	

\begin{lemma}\label{le-rf}
	For any unweighted digraph $ \calG = (V,E) $, the expected time complexity of Algorithm~\ref{alg-rf} is $ O(ln) $.
\end{lemma}

% \begin{proof}
% 	Wilson showed that the expected running time of generating a uniform spanning tree of a connected digraph $\calG$ rooted at node $u$ is equal to a weighted average of commute times between the root and the other nodes~\cite{Wi96}. Marchal rewrote this average of commute times in terms of graph matrices in Proposition 1 in~\cite{Ma00}. According to Marchal's result, the expected running time for generating one converging spanning forest is equal to the trace $\sum_{i=1}^n \omega_{ii}(1+d_i)$ of matrix $\mathbf{\Omega}(\II+\DD)$. Using Lemma~\ref{le-rf}, we have
% 	$\sum_{i=1}^n \omega_{ii}(1+d_i) \leq \sum_{i=1}^n \frac{2+2d_i}{2+d_i} \leq 2n\left(1-\frac{1}{n+1}\right).$ Thus, the expected time complexity of Algorithm 2 is $O(ln)$.
% \end{proof}

Lemma~\ref{le-rf} shows that the time complexity of algorithm~\ref{alg-rf} is closely related to the sampling number $l$. The estimation of $l$ needs the Chernoff bound~\cite{ChLu06} and is given in Theorem~\ref{th-rf}.

% Below we apply the Chernoff bound~\cite{ChLu06} to analyze the absolute error guarantee of Algorithm~\ref{alg-rf}.

% \begin{lemma}(Chernoff bound)\label{le-chernoff}
% 	Let $ x_i(1\leq i\leq l) $ be independent random variables satisfying $ |x_i- \mathbb{E}\{x_i\}|\leq M$ for all $ 1\leq i \leq l $. Let $ x = \frac{1}{l} \sum_{i=1}^l x_i $. Then we have
% 	\begin{equation}\label{key}
% 		\mathbb{P}\{|x-\mathbb{E}\{x\}| \leq \epsilon \}\geq 1-2\exp{\left(-\frac{l\epsilon^2}{2({\rm Var}\{x\}l+M\epsilon/3)}\right)}.
% 	\end{equation}
% \end{lemma}

\begin{theorem}\label{th-rf}
     For any node $i$ with $\omega_{ii} >\sigma$, and  parameters  $ \epsilon,\sigma,\delta \in(0,1) $, if $l$ is chosen obeying   $ l =\left \lceil \frac{6+2\epsilon}{3\sigma\epsilon^2}\log{\frac{2}{\delta} }  \right \rceil   $, then the approximation $\widehat{\boldsymbol{\omega}}$$[i]$ of $ \omega_{ii}$  returned by  Algorithm~\ref{alg-rf} satisfies the following relation with the  probability of at least $1-\delta$:
	\begin{equation}\label{eq-omegaii}
		(1-\epsilon) \omega_{ii}	\leq 	  \widehat{\boldsymbol{\omega}}[i] \leq (1+\epsilon) \omega_{ii}.
	\end{equation} 
\end{theorem}

From Theorem~\ref{th-rf}, it is evident that for a directed graph $G = (V,E)$, if one desires an $\epsilon$ relative error guarantee using Algorithm~\ref{alg-rf} with a probability of at least $1 - 1/n$, then the complexity of the algorithm is $O(\frac{n{\rm log}n}{\sigma\epsilon^2})$. And this guarantee is applicable only to those nodes where $\omega_{ii} > \sigma$. 
However, the situation becomes more complex for graphs containing nodes with high degrees. As shown in~\cite{SuZh23},  $\omega_{ii} $ is upper bounded by  $\frac{2}{2+d_i}$. Consequently, for nodes with significantly large degrees, the number of samples required by Algorithm~\ref{alg-rf} becomes prohibitively large, in order to achieve the desired relative error guarantee.

\section{Estimator with Reduced Variance}
 In this section, we introduce a novel estimator for $\omega_{ii}$, designed to overcome the challenge encountered in Algorithm~\ref{alg-rf}, where the number of samplings may become excessively large.

\subsection{Inspiration from FJ Model on Digraphs}
The inspiration for the new estimator is drawn from the widely recognized Friedkin-Johnsen (FJ) model~\cite{FrJo90}, a prevalent model for opinion evolution and formation. For the FJ opinion model on a digraph $\calG=(V,E)$, each node/agent  $i\in V$ is associated with two opinions: one is the internal opinion $s_i$, the other is the expressed opinion $z_i(t)$ at time $t$. The internal opinion $ s_i $, which lies within the closed interval [0, 1], represents node  $i$'s inherent stance on a specific topic. During the process of opinion evolution, the internal opinion $s_i$ remains constant, while the expressed opinion $z_i(t)$ evolves at time $ t+1$ as follows:
\begin{equation}\label{FJ}
z_i(t+1) = \frac{s_i +\sum_{j\in N(i)} z_j(t)}{1+d_i}.
\end{equation}

 Let $ \sss = (s_1,s_2,\cdots,s_n)^\top$ denote the vector of internal opinions, and let $ \zz(t) = (z_1(t),z_2(t),\cdots,z_n(t))^\top $ denote the vector of expressed opinions at time $ t $. The following lemma reveals the convergence result of the iteration.
 
\begin{lemma}\cite{BiKlOr15}\label{le-z}
Regardless of the initial value of  $\zz(0)$, if $\zz(t+1)$ evolves according to~\eqref{FJ}, where $t=1,2,\cdots,\infty,$ then as $ t $ approaches infinity, $ \zz(t) $ converges to an equilibrium vector  $ \zz = (z_1,z_2,\cdots,z_n)^\top$ satisfying $ \zz =  \mathbf{\Omega}\sss $.
\end{lemma}

Lemma~\ref{le-z} elucidates the correlation between the equilibrium expressed opinion $ \zz = (z_1,z_2,\cdots,z_n)^\top$ and the initial opinion $\sss$, with the forest matrix \( \mathbf{\Omega} \) playing a pivotal role in this relationship.   Consequently, Lemma~\ref{le-z} offers a novel approach to determine the diagonal elements of the forest matrix. Let  $\boldsymbol{\rho}^i = \mathbf{\Omega}\ee_i$  represent the $i$-th column of $\mathbf{\Omega}$. Then the $i$-th diagonal element of $\mathbf{\Omega}$ is exactly the $i$-th element ${\rho}^i_i$ of vector $\boldsymbol{\rho}^i$. To obtain the diagonal element $\omega_{ii}$, we can initially set the opinion vector $\sss = \ee_i$ and select   an appropriate vector  $\boldsymbol{\rho}^i(0)$. Then  repeat the iteration equation in~\eqref{FJ} $t$ times to yield $\boldsymbol{\rho}^i(t)$. According to Lemma~\ref{le-z}, the $ i $-th component $ {\rho}^i_i(t) $ of $ \boldsymbol{\rho}^i(t) $ serves as an estimator for $ \omega_{ii} $, and as $ t $ increases, the discrepancy between ${\rho}^i_i(t)$ and $\omega_{ii}$ diminishes.

However, employing the   iteration procedure  to compute all diagonal elements of the forest matrix presents several challenges. For a fixed node $i$, executing the iteration equation~\eqref{FJ} $t$ times needs a time complexity of $O(mt)$. Given that all $n$ diagonal elements require computation, the naive iteration approach demands  time complexity of $O(mnt)$, which is computationally expensive. 

\subsection{ Novel Unbiased Estimator}
To address this challenge, it is pertinent to note that for a specific node $i$, the required number of iterations $t $  to achieve an error bound between $\boldsymbol{\rho}^i$ and $\boldsymbol{\rho}^i(t)$ varies significantly with the initial vector $\boldsymbol{\rho}^i(0)$. Specifically, if $\boldsymbol{\rho}^i(0) $ precisely matches the equilibrium vector $\boldsymbol{\rho}^i = \mathbf{\Omega} \ee_i$, the iterative equation will maintain the value of  $\boldsymbol{\rho}^i$  unchanged for any $t = 1, \cdots$, as $ \boldsymbol{\rho}^i$  is the system's equilibrium vector.  In this scenario, the required iteration number  $t$ is zero. However, the exact value of  $\boldsymbol{\rho}^i$ is typically unknown. Adopting a similar idea, we can initialize $\boldsymbol{\rho}^i(0)$  as an easily obtainable estimator of $\boldsymbol{\rho}^i$, which brings  $\boldsymbol{\rho}^i$  and $\boldsymbol{\rho}^i(0)$  closer initially, thereby reducing the number of iterations $t$ needed.

To achieve this goal, we introduce some quantities. For a spanning converging forest $\phi\in\calF$, define a random variable $\widehat{\omega}_{ji}$ as $\widehat{\omega}_{ji}(\phi) \triangleq \mathbb{I}_{\{r_{\phi}(j)=i\}}$. The estimator $\widehat{\omega}_{ji}$ is an unbiased estimator of $\omega_{ji}$ if we randomly select $\phi$ since 
    $\mathbb{E}\{\widehat{\omega}_{ji}(\phi) \}= \mathbb{P}\{ r_{\phi}(j)=i\} = {|\calF_{ji}|}/{|\calF|} = \omega_{ij}.$
Then, using Wilson's algorithm we generate $l$ spanning converging forests $\phi_1,\cdots,\phi_l$, and set the initial iteration vector $\boldsymbol{\rho}^i(0) = ({\rho}^i_1(0),\cdots,{\rho}^i_n(0))$, where ${\rho}^i_j(0) = {\frac{1}{l}\sum_{k=1}^l \widehat{\omega}_{ji}(\phi_k)}$. After that, we repeat the iteration $t$ times and get $\boldsymbol{\rho}^i(t)$, which serves as an estimator for $\boldsymbol{\rho}^i$. 

Recall that our objective is to reduce the iteration times,  focusing on the $i$-th component of $\boldsymbol{\rho}^i$. A bold and natural idea emerges: What if we only perform one iteration? This case offers a unique perspective that might be the key to solving the challenge. In this scenario,  the estimator  $\boldsymbol{\rho}^i(t)$  can be elegantly expressed as: $\boldsymbol{\rho}^i(1) = (\II+\DD)^{-1}\ee_i+(\II+\DD)^{-1}\AA\boldsymbol{\rho}^i(0)$. Moreover, since no further iterations are needed, we only need to focus on the $i$-th component of $\boldsymbol{\rho}^i(1)$, and do not need to calculate the other elements. That is, we only need to calculate that $ \rho^i_i(1) = \frac{1}{1+d_i}(1+\sum_{j\in N(i)}\rho^i_j(0)) =  \frac{1}{1+d_i}\big(1+\frac{1}{l}\sum_{j\in N(i)}{\sum_{k=1}^l \widehat{\omega}_{ji}(\phi_k)}\big).$

From the expression of $\rho^i_i(1)$, we  define a new estimator as $   \widetilde{\omega}_{ii}(\phi) \triangleq \frac{1}{1+d_i}\big(1 + \sum_{j\in N(i)}\widehat{\omega}_{ji}(\phi)\big).$ Then we derive that $\rho^i_i(1) = \frac{1}{l}\sum_{k=1}^l  \widetilde{\omega}_{ii}(\phi_k)$. That is, we can directly use Wilson's algorithm to sample $l$ spanning converging forests, and then obtain the value of $\widetilde{\omega}_{ii}(\phi_k)$. The average of $l$ values is equal to the $i$-th component of $\boldsymbol{\rho}^i(1)$. In Algorithm~\ref{alg-rfwithvar}, we provide the details of this method to obtain the novel variable. This algorithm takes a parameter $l$,  the number of spanning converging forests to be sampled, and then returns a vector  $\widetilde{\boldsymbol{\omega}}$  as an estimator for the diagonal elements of the forest matrix.

Algorithm~\ref{alg-rfwithvar} starts by initializing $\widetilde{\omega}_{ii}(\phi)$  to $\frac{1}{1+d_i}$, then  generates  $l$  spanning converging forests using Wilson's Algorithm. After generating each forest $\phi_t$, a loop is executed to update the vector $\widetilde{\boldsymbol{\omega}}$. Similar to the analysis of Lemma~\ref{le-rf}, the total computational complexity of Algorithm~\ref{alg-rfwithvar} is $O(ln)$. Lemma~\ref{le-var-omega} shows that $\widetilde{\omega}_{ii}(\phi_k)$ is still an unbiased estimator of $\omega_{ii}$  but has  less variance than $\widehat{\omega}_{ii}$.

\begin{algorithm} 
	\caption{$\textsc{SCFV}(\calG,l)$}
	\label{alg-rfwithvar}
	\Input{ $\calG$:a digraph  \\ $l$:number of generated spanning converging forest
	}
	\Output{$\widetilde{\boldsymbol{\omega}}$ : a vector approximating the diagonal elements of the forest matrix}
 \textbf{Initialize} :
	$\widetilde{\boldsymbol{\omega}}$[$i$] $ \leftarrow \frac{1}{1+d_i} $, $i= 1,2,\ldots,n $\\
 \For{$t = 1$ to $l$}{
 $\phi_t$ $\leftarrow$ a spanning converging forest generated from $\calG$\\
 \For{$ i = 1 $ to $ n $}{
 $j \leftarrow r_{\phi_t}(i)$\\
                 \If{$i \in N({j})$}{
                $\widetilde{\boldsymbol{\omega}}$[$j$]$\leftarrow$ $\widetilde{\boldsymbol{\omega}}$[$j$] $ + \frac{1}{l(1+d_j)}$\\

                }
 }
 
	 }
\textbf{return} $\widetilde{\boldsymbol{\omega}}$ \;
\end{algorithm}

\begin{lemma}\label{le-var-omega}
       For any node $i\in V$,   $\widetilde{\omega}_{ii}(\phi)$ is an unbiased estimator of $\omega_{ii}$.  Let $q^{(i)}_{jk}$ represent the ratio of the number of spanning converging forests where both roots of $j$ and $k$ are node $i$  to the total number of spanning converging forests. Then, the variance of this estimator is   ${\rm Var}\{\widetilde{\omega}_{ii}(\phi)\} = \frac{3\omega_{ii}}{1+d_i} - \frac{2}{(1+d_i)^2} + \frac{2\sum_{j,k\in N(i)}q^{(i)}_{jk}}{(1+d_i)^2}-\omega_{ii}^2,$  which is always less than or equal to the variance of the estimator $\widehat{\omega}_{ii}$. 
\end{lemma}

Lemma~\ref{le-var-omega} highlights the reduced variance of the random variable $\widetilde{\omega}_{ii}(\phi)$. However, a challenge arises when we invoke the Chernoff bound to determine the requisite sampling number. Since the third term of the variance  is inherently complex, deriving a proper upper bound for this term is not straightforward, which complicates the task of determining an appropriate sampling number $l$. Consequently, although the new estimator  $\widetilde{\omega}_{ii}(\phi) $ facilitates a reduction in variance, we are unable to obtain a satisfying theoretical result for the sampling number $ l $, due to the complex form of the third term  for ${\rm Var}\{\widetilde{\omega}_{ii}(\phi)\}$.

% Specifically, for the sampling number to be independent of the lower bound of $ \omega_{ii} $ for all nodes $ i \in V $, from the inequality~\eqref{ltomeet}, it is desirable that the ratio $ {{\rm Var}\{\widetilde{\omega}_{ii}\}}/{\omega_{ii}^2} $ remains bounded by a constant $ K $. However, this aspiration faces a challenge when we examine the third term of the variance, namely $  \frac{2\sum_{j,k\in N(i)}q^{(i)}_{jk}}{(1+d_i)^2\omega_{ii}^2}  $. Estimating this term is not straightforward, which complicates the task of determine an appropriate sampling numbers $l$. 

% Consequently, while the reduced variance of $ \widetilde{\omega}_{ii}(\phi) $ is promising, the complexities introduced by this term need further investigation. 

% A closer look reveals that the third term of the variance for $ \widetilde{\omega}_{ii}(\phi) $  arises from the relationship that for distinct nodes  $i,j,k \in V$, we have  $\mathbb{E}\{\widehat{\omega}_{ji}(\phi)\widehat{\omega}_{ki}(\phi)\} = q_{jk}^{(i)}$ . This represents the proportion of scenarios where nodes  $j$  and  $k$  both have  $i$  as their root across all spanning converging forests. It's worth noting that in the spanning converging forest $\phi$, although two nodes might share the same root, any given node can have only one root. This observation provides a spark of inspiration. By leveraging this unique characteristic, we can craft a new estimator. This estimator would be devoid of the cross-product term that complicates our variance, yet retain the other attributes of  $\widetilde{\omega}_{ii}(\phi)$.

\section{ New Iteration Equation for Further Variance Reduction}
In this section, we introduce a new iteration equation and propose a superior estimator.  This novel estimator  notably omits the complex cross-product term  in the variance of $\widetilde{\omega}_{ii}(\phi)$, which further reduces the variance  in samplings, and allows to derive a better theoretical result.

\subsection{ New Iteration Equation}
To further refine our estimator, it is insightful to revisit the estimator $\widetilde{\omega}_{ii}(\phi)$, which draws inspiration from \eqref{FJ}. In \eqref{FJ}, the update of expressed opinion for each node is influenced by its neighbors. Specifically, $z_i(t+1)$ is updated according to the value of $z_j(t)$ where $j\in N(i)$. A novel idea emerges when we consider inverting the direction of opinion dissemination, implying that $z_i(t+1)$ is updated according to the value of $z_j(t)$ where $i\in N(j)$. In this scenario, our focus shifts to the $i$-th row of the forest matrix $\mathbf{\Omega}$, denoted by ${\boldsymbol{\gamma}^i}^\top  \triangleq \ee^\top_i \mathbf{\Omega}$. The $i$-th component of ${\boldsymbol{\gamma}^i}$ precisely corresponds to the $i$-th diagonal element of the forest matrix. To compute ${\boldsymbol{\gamma}^i}$, we propose an iterative equation similar to \eqref{FJ}, as follows:
\begin{equation}\label{FJ2}
    {\boldsymbol{\gamma}^i}^\top(t+1) = \ee_i^\top(\II+\DD)^{-1} +{\boldsymbol{\gamma}^i}^\top(t)\AA(\II+\DD)^{-1}.
\end{equation}
Since matrix  $\II+\LL$  is invertible, it is the same with  matrix $\II+\LL^\top$. Consequently, the iteration of ${\boldsymbol{\gamma}^i}^\top(t)$ in \eqref{FJ2} converges to ${\boldsymbol{\gamma}^i}^\top$. Consistent with the previous analysis, we aim to restrict the process to a single iteration, as multiple iterations are time unaffordable.

To attain this goal, we continue to leverage Wilson's algorithm to obtain ${\boldsymbol{\gamma}^i(0)}$, which serves as a preliminary estimation of the $i$-th row of the forest matrix. Specifically, we employ Wilson's algorithm to generate $l$ spanning converging forests $\phi_1,\cdots,\phi_l$, and set the initial iteration vector $\boldsymbol{\gamma}^i(0) = ({\gamma}^i_1(0),\cdots,{\gamma}^i_n(0))$, where ${\gamma}^i_j(0) = {\frac{1}{l}\sum_{k=1}^l \widehat{\omega}_{ij}(\phi_k)}$. Upon performing one time of iteration according to \eqref{FJ2}, we derive that 
$$ \gamma^i_i(1) = \frac{1}{1+d_i}\big(1+\sum_{j: i\in N(j)}\gamma^i_j(0)\big) =  \frac{1}{1+d_i}\big(1+\frac{1}{l}\sum_{j: i\in N(j)}{\sum_{k=1}^l \widehat{\omega}_{ij}(\phi_k)}\big).$$

% We obtain:
% \begin{equation}\label{Omega0}
%     \mathbf{\Omega} = (\II+\DD)^{-1} + \PP\mathbf{\Omega}. 
% \end{equation}

% To derive a new formulation, we first multiply both sides of Equation \eqref{Omega0} by  $(\II+\LL)$  on the right-hand side and  $(\II+\DD)$  on the left-hand side. Subsequently, we also multiply  $(\II+\LL)$  on the left-hand side and  $(\II+\DD)$  on the right-hand side of the equation. This mathematical manipulation leads us to Equation \eqref{Omega0new}, providing the foundation for our novel estimator.

% \begin{equation}\label{Omega0new}
%     \mathbf{\Omega} = (\II+\DD)^{-1} + \mathbf{\Omega}\AA(\II+\DD)^{-1}. 
% \end{equation}

With this new formulation, we  define a novel random variable, $\widebar{\omega}_{ii}(\phi)$, for any spanning converging forest  $\phi \in \calF $ and $i\in V$  as 
$      \widebar{\omega}_{ii}(\phi) \triangleq \frac{1}{1+d_i} + \frac{1}{1+d_i}\sum_{j:i\in N(j)}\widehat{\omega}_{ij}(\phi).   
$

Then, it follows that $\gamma^i_i(1) = \frac{1}{l}\sum_{k=1}^l  \widebar{\omega}_{ii}(\phi_k)$. Thus, in order to obtain $\gamma^i_i(1)$,  we can directly employ Wilson's algorithm to sample $l$ spanning converging forests, and obtain the value of $\widebar{\omega}_{ii}(\phi_k)$. The average of these $l$ values is equal to the $i$-th component of $\boldsymbol{\gamma}^i(1)$. We detail this in Algorithm~\ref{alg-rfwithvar+}, which demonstrates the method to derive this novel variable. The time complexity of Algorithm~\ref{alg-rfwithvar+} is  $O(ln)$, where $l$ represents the number of spanning converging forests sampled.

%%%%%%%%%%%%%%%%%%%%%%%%%%%%%%%%%%%%%%%%%%%%%%%	
% Algorithm 2
%%%%%%%%%%%%%%%%%%%%%%%%%%%%%%%%%%%%%%%%%%%%%%%	
% \begin{small}
\begin{algorithm}
	\caption{$\textsc{SCFV+}(\calG,l)$}
	\label{alg-rfwithvar+}
	\Input{ $\calG$ : a digraph  \\ $l$:number of generated spanning converging forest
	}
	\Output{$\widebar{\boldsymbol{\omega}}$ : a vector approximating the diagonal elements of the forest matrix}
 \textbf{Initialize} :
	$\widebar{\boldsymbol{\omega}}$[$i$] $ \leftarrow \frac{1}{1+d_i} $, $i= 1,2,\ldots,n $\\
 \For{$t = 1$ to $l$}{
 $\phi_t$ $\leftarrow$ a spanning converging forest generated from $\calG$\\
 \For{$ i = 1 $ to $ n $}{
 $j \leftarrow r_{\phi_t}(i)$\\
                 \If{$i \in N({j})$}{
                $\widebar{\boldsymbol{\omega}}$[$i$]$\leftarrow$ $\widebar{\boldsymbol{\omega}}$[$i$] $ + \frac{1}{l(1+d_i)}$\\

                }
 }
% InForest[$ i $] $ \leftarrow $ false , $i= 1,2,\ldots,n $\;
% 	Next[$ i $] $ \leftarrow -1$ , $i= 1,2,\ldots,n $\;
%         RootIndex[$i$] $ \leftarrow 0 $, $i= 1,2,\ldots,n $\\
% 	\For{$ i = 1 $ to $ n $}
% 	{$ u \leftarrow i $\; 
% 		\While{not InForest[$ u $]}{
% 			seed $ \leftarrow $ \textsc{Rand}()  \; 
% 			\If{seed  $\leq  \frac{1}{1+d_u} $}{
% 				InForest[$ u $] $ \leftarrow $ true\;
% 				Next[$ u $]$ \leftarrow -1 $\;
%     RootIndex[$ u $]$ \leftarrow u $\;
% }
% 			\Else{
% 				Next[$ u $] $ \leftarrow $ \textsc{RandomSuccessor}($ u,\calG $)\; 
% 				u $ \leftarrow $ Next[$ u $]\;
% 			}
% 		}
%             RootNow $\leftarrow  $ RootIndex[$u$]\;
%             $ u \leftarrow i $\\
% 		\While{not InForest[$ u $]}{
% 			InForest[$ u $] $ \leftarrow $ true\;
%                 \If{$u \in N(\rm{Rootnow})$}{
%                 $\widebar{\boldsymbol{\omega}}$[$u$]$\leftarrow$ $\widebar{\boldsymbol{\omega}}$[$u$] $ + \frac{1}{l(1+d_u)}$\\

%                 }
                
% 			u $ \leftarrow $ Next[$ u $]\;
% 		}
% 	}
	 }
\textbf{return} $\widebar{\boldsymbol{\omega}}$ \;
\end{algorithm}
% \end{small}
%%%%%%%%%%%%%%%%%%%%%%%%%%%%%%%%%%%%%%%%%%%%%%%	

% \begin{equation}
% \begin{aligned}   
%       \widebar{\omega}_{ii}(\phi) &= \ee_{i}^\top(\II+\DD)^{-1} \ee_{i}+ \ee_{i}^\top\mathbf{\Omega}\AA(\II+\DD)^{-1}\ee_{i} \\  &= \frac{1}{1+d_i} + \frac{1}{1+d_i}\sum_{j:i\in N(j)}\widehat{\omega}_{ij}(\phi).   
% \end{aligned}
% \end{equation}

\subsection{Advanced Estimator Analysis and Implementation}
After having proposed a novel estimator $\widebar{\omega}_{ii}$ for $\omega_{ii}$, we  show that this new estimator $\widebar{\omega}_{ii}(\phi)$ is superior to $\widetilde{\omega}_{ii}(\phi)$.

\begin{lemma}\label{le-var-omega-new}
    For any node  $i \in V$, $\widebar{\omega}_{ii}(\phi)$  serves as an unbiased estimator of  $\omega_{ii} $. The variance of this estimator is 
    ${\rm Var}\{\widebar{\omega}_{ii}(\phi)\} = \frac{3\omega_{ii}}{1+d_i} - \frac{2}{(1+d_i)^2} -\omega_{ii}^2,$ which is consistently less than or equal to the variance of the estimator  $\widetilde{\omega}_{ii}(\phi)$, as well as the variance of the estimator $ \widehat{\omega}_{ii}(\phi) $.
\end{lemma}

In contrast to $\widetilde{\omega}_{ii}$, the variance of $\widebar{\omega}_{ii}$ does not include the complex cross-product term, which previously posed significant challenges when deriving its upper bound. The absence of this complex term in the new estimator simplifies our analysis, as highlighted in  Lemma~\ref{le-varbound}. This lemma illuminates a significant attribute of the estimator $\widebar{\omega}_{ii}$. Specifically, it delineates an upper bound for the ratio of the variance of $\widebar{\omega}_{ii}$ to the square of $\omega_{ii}$, a crucial factor when determining the sampling number.

\begin{lemma}\label{le-varbound}
    Given a directed graph $\calG = (V,E)$, for any node $i\in V$, the ratio $ \frac{{\rm Var}\{\widebar{\omega}_{ii}\}}{\omega_{ii}^2}$ is constrained by $\frac{1}{8}$. Formally, $ \frac{{\rm Var}\{\widebar{\omega}_{ii}\}}{\omega_{ii}^2}\leq \frac{1}{8}$.
\end{lemma}

% \begin{proof}
%     According to Lemma~\ref{le-var-omega-new}, we derive that
%     \begin{equation}
%     \begin{aligned}
%               &|\frac{{\rm Var}\{\widebar{\omega}_{ii}\}}{\omega_{ii}^2}| = \frac{3}{(1+d_i)\omega_{ii}} - \frac{2}{(1+d_i)^2\omega_{ii}^2} -1 \\&= -\frac{2}{(1+d_i)^2}(\frac{1}{\omega_{ii}}-\frac{3(1+d_i)}{4})^2+\frac{1}{8}  \leq \frac{1}{8}
%     \end{aligned}
%     \end{equation} 
% The equation holds when $\frac{1}{\omega_{ii}}=\frac{3(1+d_i)}{4}$,  implying $\omega_{ii} = \frac{4}{3(1+d_i)}$, which finishes the proof.
% \end{proof}

Armed with the insights from Lemma~\ref{le-varbound}, we proceed  to  establish  a connection between the relative error bound and the number of samples in Algorithm~\ref{alg-rfwithvar+}. 

\begin{theorem}\label{th-Fast}
For any  $ \epsilon\in(0,1) $  and $ \delta\in(0,1) $, if $l$ is chosen obeying   $ l =\left \lceil  (\frac{2}{3\epsilon}+\frac{1}{4\epsilon^2})\log(\frac{2}{\delta})  \right \rceil   $, then the approximation $ \widebar{\boldsymbol{\omega}}[i]$ of $ \omega_{ii}$ returned by Algorithm~\ref{alg-rfwithvar+} satisfies the following relation with probability of at least $1-\delta$:  $(1-\epsilon)\omega_{ii}	\leq 	  \widebar{\boldsymbol{\omega}}[i] \leq (1+\epsilon) \omega_{ii}$.
	% \begin{equation}\label{key}
	% 	(1-\epsilon)\omega_{ii}	\leq 	  \widebar{\boldsymbol{\omega}}[i] \leq (1+\epsilon) \omega_{ii}.
	% \end{equation} 
\end{theorem}

Based on Theorem~\ref{alg-rfwithvar+}, when we fix the failure probability $\delta$ and the relative error parameter $\epsilon$,  the required number of samples remains invariant regardless of the size and structure of the graph. This theoretical insight highlights the superiority of the estimator $\widebar{\omega}$ over both $\widehat{\omega}$ and $\widetilde{\omega}$.

\section{Experiments}
% In this section, we conduct extensive experiments on various real-life networks in order to evaluate the performance of our three algorithms \textsc{SCF}, \textsc{SCFV}, and \textsc{SCFV+}, in terms of effectiveness and efficiency. 

%Our source code is publicly available on \url{https://anonymous.4open.science/r/Diagonal-of-Forest-Matrix}.
\subsection{Setup}
 \textbf{Datasets and Equipment.}
 The datasets of selected real networks are publicly available in the KONECT~\cite{Ku13} and SNAP~\cite{LeSo16}. Our experiments are conducted on a diverse range of both undirected and directed networks, including  social networks, technical networks, and so on. For these datasets, the number $n$ of nodes ranges from about 16 thousand to 33 million, and the number $m$ of directed edges ranges from about 25 thousand to 301 million. The details of these datasets are presented in Table~\ref{datasets}. All experiments are conducted using the Julia programming language on a single-threaded setup. We conduct all experiments in a computational environment featuring a 4.2 GHz Intel i7-7700 CPU with 64GB of primary memory.

%The source code is publicly available on \url{https://github.com/OpinionOptimization/Opinion}.
\begin{table}[b]\fontsize{8}{11}
\caption{Datasets used in experiments. }
\begin{tabular}{cccc}
\hline
\textbf{Type}                                                                & \textbf{Dataset}    & $n$        & $m$         \\ \hline
\multirow{9}{*}{\begin{tabular}[c]{@{}c@{}}undirected\\ graphs\end{tabular}} & web-webbase-2001    & 16,062     & 25,593      \\
& soc-gemsec-RO       & 41,773     & 125,826     \\& tech-p2p-gnutella   & 62,561     & 147,878     \\& tech-RL-caida       & 190,914    & 607,610     \\& soc-twitter-follows & 404,719    & 713,319     \\& soc-delicious       & 536,108    & 1,375,961   \\& dblp                & 5,624,219  & 12,282,055  \\& livejournal         & 7,489,073  & 112,307,315 \\& delicious           & 33,777,767 & 301,183,342 \\ \hline
\multirow{9}{*}{\begin{tabular}[c]{@{}c@{}}directed\\ graphs\end{tabular}}  & wikipedialinks        & 17,649     & 296,918     \\ & p2p-gnutella31      & 62,586     & 147,892     \\
& email-euall         & 265,009    & 418,956     \\& web-Stanford        & 281,903    & 2,312,500   \\& web-Google          & 875,713    & 5,105,039 
   \\& northwestUSA        & 1,207,945   & 2,820,774    \\& wikitalk            & 2,394,385 
  & 5,021,410    \\& greatlakes          & 2,758,119 
  & 6,794,808    \\& fullUSA             & 23,947,347 & 57,708,624 \\ \hline
\end{tabular}\label{datasets}
\end{table}
\noindent\textbf{Algorithms.}
For evaluating the performance of our algorithms estimating  the diagonal elements of the forest matrix, we compare our three proposed algorithms \textsc{SCF}, \textsc{SCFV}, and \textsc{SCFV+} with two state-of-the-art algorithms, namely \textsc{JLT} \cite{JiBaZh19} and \textsc{UST} \cite{GrAnPrMe21}.  While the two state-of-the-art algorithms, \textsc{JLT} and \textsc{UST}, are confined to undirected graphs due to the limitations of the Laplacian solver they employ, our proposed algorithms \textsc{SCF}, \textsc{SCFV}, and \textsc{SCFV+}   are applicable to both undirected graphs and digraphs.
% The \textsc{JLT} algorithm combines the Johnson-Lindenstrauss lemma \cite{JoLi84,Ac03} with the  fast Laplacian solver \cite{CoKyMiPaJaPeRaXu14}. 
%  \textsc{JLT} needs  time ofm$O(m\epsilon^{-2}\log^{2.5}n\log\frac{1}{\epsilon}{\rm polyloglog}(n))$ to achieve a relative error bound. 
% On the other hand, the \textsc{UST} algorithm employs a single instance of a Laplacian solver and utilizes Wilson's algorithm to sample uniform spanning trees, thereby deriving the effective resistance. The author in \cite{GrAnPrMe21} elucidates the correlation between the diagonal elements of the forest matrix and the effective resistance in the augmented graph. The \textsc{UST} algorithm incurs a total time complexity of \(\widetilde{O}(m\epsilon^{-2}\log^{3/2}n)\) to guarantee an absolute error of \(\epsilon\) with high probability.

\subsection{Experiments on Undirected Graphs}
In this subsection, an analysis is conducted to compare our proposed algorithms \textsc{SCF}, \textsc{SCFV}, and \textsc{SCFV+} and two state-of-the-art algorithms \textsc{JLT} and \textsc{UST}, on undirected graphs. 

We first evaluate the accuracy of the algorithms. To this end, we perform experiments on six small or medium-sized networks: web-webbase-2001, soc-gemsec-RO, tech-p2p-gnutella, tech-RL-caida, soc-twitter-follows, and soc-delicious.
For first three graphs, the ground truth of the diagonal elements of $\mathbf{\Omega}$ is computed by directly inverting the matrix $\II+\LL$. For the last three graphs, the Conjugate Gradient  solver with a tolerance of $10^{-9}$ is utilized to obtain the diagonal, since direct matrix inversion is computationally infeasible. This approach aligns with the settings in~\cite{GrAnPrMe21}.

%For ease of discussion, we employ lowercase letters a through f to denote the six networks: web-webbase-2001(a), soc-gemsec-RO(b), tech-p2p-gnutella(c), tech-RL-caida(d), soc-twitter-follows(e), and soc-delicious(f). 

Two metrics are used to evaluate the accuracy: average relative error for all nodes and maximum relative errors among all nodes. In addition, there are two parameters in the five considered algorithms: the sampling number $l$ in \textsc{UST}, \textsc{SCF}, \textsc{SCFV}, and \textsc{SCFV+}, and the dimension $k$ for the Johnson-Lindenstrauss  lemma in \textsc{JLT}. We set $l$  to 500 for the four sampling-based algorithms, and set $k$  to 50 for \textsc{JLT}. The results for these settings are reported in Figure~\ref{f1}.

\begin{figure}[htbp!]
	\centering
	\includegraphics[width=.95\columnwidth]{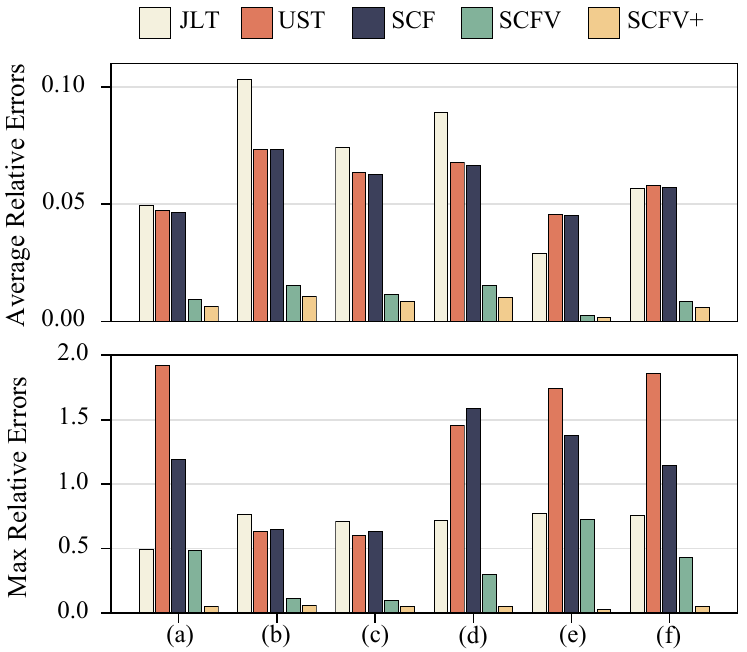}
	\caption{Comparison of maximum  and average   relative errors for five algorithms on six undirected graphs: web-webbase-2001 (a), soc-gemsec-RO (b), tech-p2p-gnutella (c), tech-RL-caida (d), soc-twitter-follows (e), and soc-delicious (f). }\label{f1}	
\end{figure}

The upper panel of Figure~\ref{f1} depicts the average relative errors for the five algorithms on six networks, while the lower panel illustrates the maximum relative errors. With regard to the average relative error, it is evident that \textsc{JLT}, \textsc{UST}, and \textsc{SCF} yield similar results. In contrast, \textsc{SCFV} and \textsc{SCFV+} achieve nearly 10 times better accuracy  compared to \textsc{UST} and \textsc{SCF}, with \textsc{SCFV+} securing the best result due to its minimized variance. As for the maximum relative error, it is noteworthy that, \textsc{SCFV+} consistently delivers stable and satisfactory results across the six graphs. The other algorithms, to varying degrees, yield results that may be deemed less accurate.

We continue  to delve deeper into the effectiveness and efficiency of the five algorithms with varying parameters. For this purpose, we vary  parameters $l$ and $k$, and observe how the maximum and average errors change with these parameters. In our experiments, we set \(l = 500, 1000, 2000\) and \(k = 50, 100, 150\), and report  their resultant effects in Figure~\ref{f2}, which are displayed by scatter plots.

% \begin{figure}[htbp!]
% 	\centering
% 	\includegraphics[width=1\columnwidth]{dotaveragesmall}
% 	\caption{Scatter plot of average relative errors over time for five algorithms on six undirected graphs, considering three parameters each}\label{f2}	
% \end{figure}

\begin{figure}[htbp!]
    \centering
    \begin{minipage}[]{1\columnwidth}
        \includegraphics[width=1.0\linewidth]{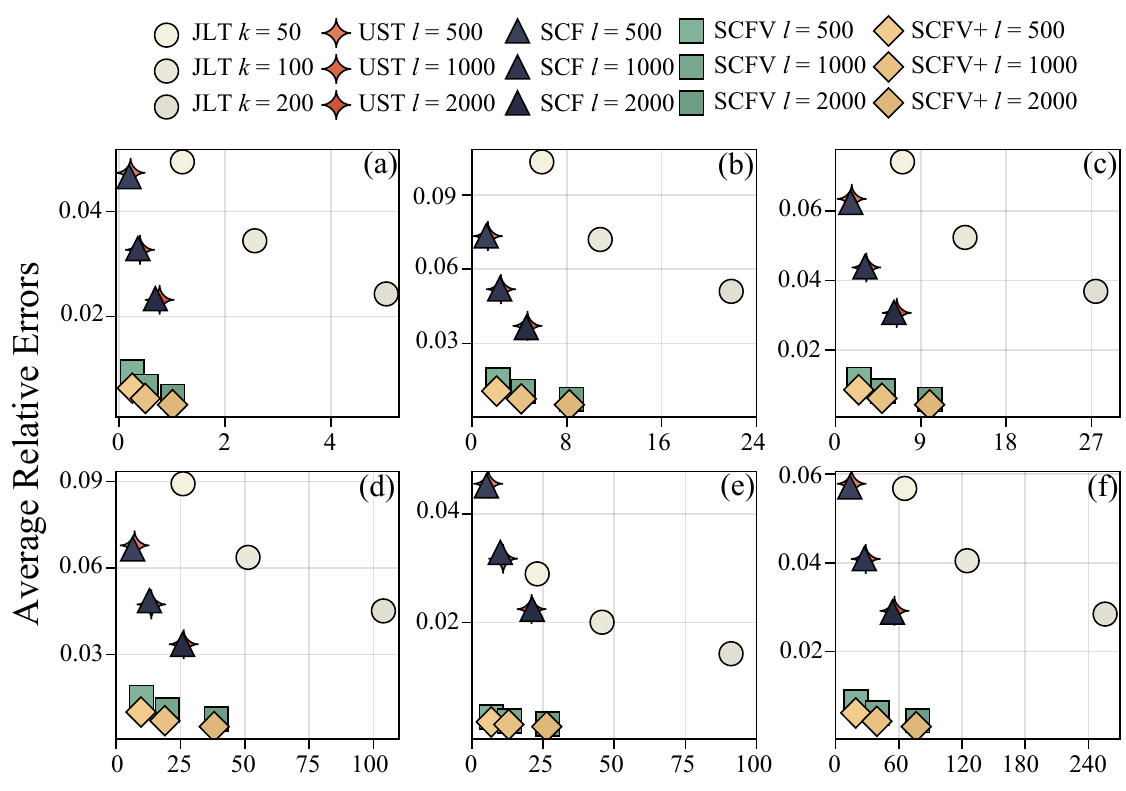}
    \end{minipage}    
 
    \begin{minipage}[]{1\columnwidth}
        \includegraphics[width=1\linewidth]{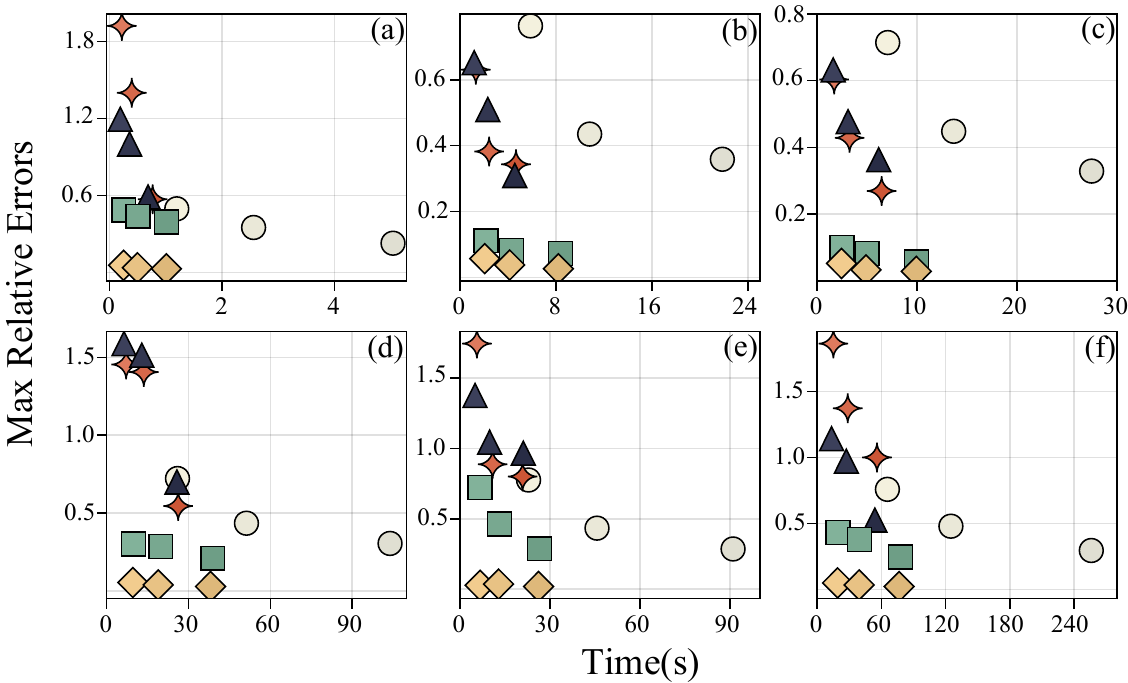}
    \end{minipage}   
    \caption{Scatter plot of maximum  and average  relative errors over time for five algorithms on six undirected graphs, considering three parameters each.}
    \label{f2}	
\end{figure}

From Figure~\ref{f2} we can see that as  parameters $l$ and $k$ increase, all algorithms require more time but produce better results. Among the five algorithms, JLT shows the least ideal performance because its time requirement increases rapidly and its error is unsatisfactory regardless of the average relative error or the maximum relative error.

%, whether considering average relative errors or maximum relative errors. 

%The remaining four algorithms all rely on sampling using Wilson's algorithm. If we fix the number of samples, \textsc{SCF} operates slightly faster than \textsc{UST}, since \textsc{UST} requires performing one instance of a fast Laplacian solver while \textsc{SCF} does not. Moreover, \textsc{SCF} achieves results comparable to \textsc{UST}. When the  sample number is identical,  \textsc{SCFV} and \textsc{SCFV+} require slightly more time than \textsc{UST} and \textsc{SCF}, while the results returned by them are superior. 

The remaining four algorithms all involve sampling using Wilson's algorithm. If we fix the number of samples, \textsc{SCF} runs slightly faster than \textsc{UST} because \textsc{UST} requires a single execution of the fast Laplacian solver while \textsc{SCF} does not. Despite the shorter running time, \textsc{SCF} achieves comparable results to \textsc{UST}. Under the condition that the number of forest samplings is identical , \textsc{SCFV} and \textsc{SCFV+} require slightly longer running time than \textsc{UST} and \textsc{SCF}, but the results returned by \textsc{SCFV} and \textsc{SCFV+} are significantly better than the corresponding results of \textsc{UST} and \textsc{SCF}.

% \begin{figure}[tbp!]
% 	\centering
% 	\includegraphics[width=1\columnwidth]{dotmax}
% 	\caption{Scatter plot of maximum relative errors over time for five algorithms on six undirected graphs, considering three parameters each}\label{f3}	
% \end{figure}

It is obvious from Figure~\ref{f2} that among the five algorithms, \textsc{SCFV+} has the best performance. Setting \(l=500\) in \textsc{SCFV+} can achieve lower errors and faster speeds than the other three sampling algorithms with \(l=2000\) and \textsc{JLT} with \(k=200\). Thus, \textsc{SCFV+} outperforms the other four algorithms, in terms of the average relative error and the maximum relative error, and the running time, when the parameters are chosen properly.

%whether considering average relative errors or maximum relative errors, and demonstrates stability and satisfactory results.

%In the subsequent analysis, we focus on three larger undirected graphs: dblp, livejournal, and delicious. As outlined in Table~\ref{datasets}, these graphs are notably substantial, each featuring over 5 million nodes and surpassing 12 million edges. Given time and storage constraints, obtaining the ground truth answer is unfeasible. Similarly,  algorithms \textsc{JLT} and \textsc{UST} fail to run on our equipment due to the time and storage limitations associated with the Laplacian solver, which is utilized in both \textsc{JLT} and \textsc{UST}. Consequently, we perform our three sampling algorithms \textsc{SCF}, \textsc{SCFV} and \textsc{SCFV+}, recording the running time, with the results tabulated in Table~\ref{tb2}. 

We finally present experimental results on three large undirected graphs: dblp, livejournal, and delicious, each of which has over 5 million nodes and over 12 million edges.  Due to time and storage constraints, obtaining the ground truth answer for these three graphs is infeasible.  Similarly, algorithms \textsc{JLT} and \textsc{UST} fail to run, because of the large time and memory requirements for running the Laplacian solver recalled by \textsc{JLT} and \textsc{UST}. Thus, we only perform our three sampling algorithms \textsc{SCF}, \textsc{SCFV} and \textsc{SCFV+}, record their running time, and report the results in Table~\ref{tb2}.
The results indicate that our three algorithms show good scalability, which  are suitable for massive networks with more than ten million nodes.

\begin{table}[htbp!]\fontsize{8}{11}
\setlength{\tabcolsep}{1.2mm} 
\caption{Running time of \textsc{SCF}, \textsc{SCFV} and \textsc{SCFV+} on three large undirected networks. }
\begin{tabular}{ccccccccc}
\hline
\multirow{3}{*}{\textbf{Graph}} &  & \multicolumn{7}{c}{\textbf{Time (seconds)}}                          \\ \cline{3-9} 
  &  & \multicolumn{3}{c}{$l=500$} &  & \multicolumn{3}{c}{$l=1000$} \\ \cline{3-5} \cline{7-9} 
 &  & \textsc{SCF}     & \textsc{SCFV}    & \textsc{SCFV+}   &  & \textsc{SCF}     & \textsc{SCFV}    & \textsc{SCFV+}    \\ \hline
dblp                             &  & 377     & 512     & 518     &  & 765     & 1034    & 1060     \\
livejournal                      &  & 321     & 508     & 541     &  & 669     & 1030    & 1051     \\
delicious                        &  & 1649    & 2437    & 2513    &  & 3262    & 4872    & 4898     \\ \hline
\end{tabular}\label{tb2}
\end{table}

%An examination of Table~\ref{tb2} indicates that our three algorithms exhibit admirable scalability, performing adeptly on the three expansive networks. 

%\section{Leader-Follower Opinion Dynamics Model and Interpretation  in Terms of  Spanning  Forests}

\subsection{Experiments on Digraphs}
We now present the results of experiments on  directed networks, aiming to evaluate the three proposed algorithms: \textsc{SCF}, \textsc{SCFV}, and \textsc{SCFV+}. Since \textsc{JLT} and \textsc{UST} do not apply to digraphs, we  exclude them from comparison.
 For the determination of the ground truth of the diagonal elements of $\mathbf{\Omega}$, the GMRES algorithm~\cite{SaSc86} is employed with a tolerance set to \(10^{-9}\).
\begin{figure}[htbp!]
	\centering
	\includegraphics[width=1\columnwidth]{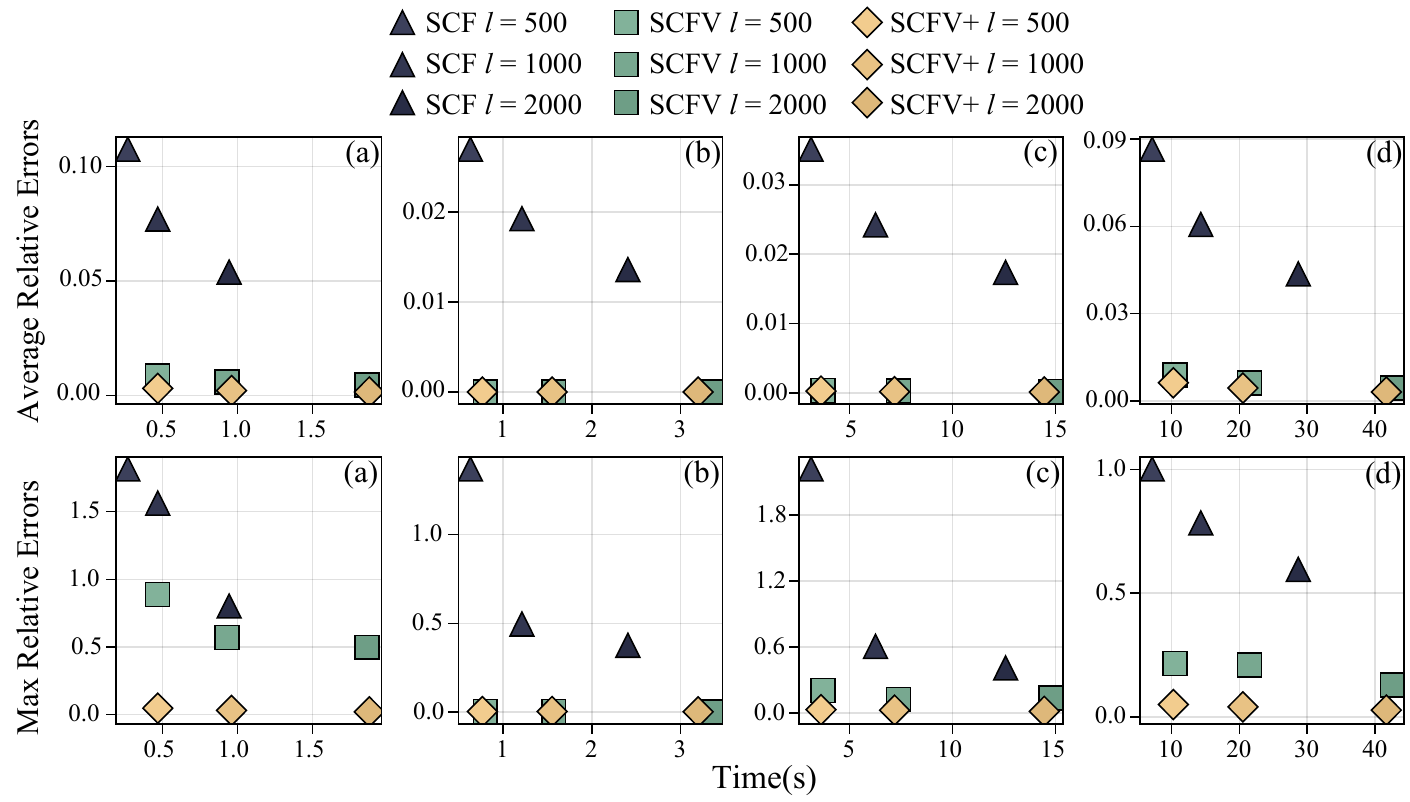}
	\caption{Scatter plot of maximum  and average  relative errors over time for three algorithms on four digraphs: wikipedialinks (a), p2p-gnutella31 (b), email-euall (c), and web-Stanford (d).}\label{f4}	
\end{figure}

We first demonstrate the results for four small or medium-sized digraphs in Figure~\ref{f4}.
 We can see that  for the three algorithms \textsc{SCF}, \textsc{SCFV}, and \textsc{SCFV+}, their errors  decrease as $l$ increases.
Moreover,  these three algorithms exhibit progressive results, but \textsc{SCFV+} outperforms the other two in both average and maximum relative errors. While both \textsc{SCFV}, and \textsc{SCFV+} yield satisfactory results in terms of average relative errors, only \textsc{SCFV+} demonstrates reliable  performance for maximum relative errors.  This is consistent with the theoretical results in Lemmas~\ref{le-var-omega} and~\ref{le-var-omega-new}.   %that the variance of \textsc{SCFV+} excludes the cross-product term present in the variance of \textsc{SCFV}, as outlined . 
Notably, \textsc{SCFV+} with $l=500$ ourperforms \textsc{SCF} and \textsc{SCFV} with $l=2000$, aligning with experimental results obtained for undirected graphs.

We next present experimental results  on five large digraphs. For these digraphs, we cannot obtain the ground truth for the diagonal, due to the time and storage limitations. In contrast, our three proposed algorithms still work for these digraphs, and demonstrate commendable scalability and efficacy. For example, for network fullUSA with more than 23 million nodes and 57 million edges,  our three algorithms return results within  12 minutes for $l=500$.

\begin{table}[htbp!]\fontsize{8}{11}\setlength{\tabcolsep}{1.2mm} 
\caption{Running time of \textsc{SCF}, \textsc{SCFV} and \textsc{SCFV+} on five large directed networks. }
\begin{tabular}{ccccccccc}
\hline
\multirow{3}{*}{\textbf{Graph}} &  & \multicolumn{7}{c}{\textbf{Time (seconds)}}                          \\ \cline{3-9} 
                                 &  & \multicolumn{3}{c}{$l=500$} &  & \multicolumn{3}{c}{$l=1000$} \\ \cline{3-5} \cline{7-9} 
                                 &  & \textsc{SCF}    & \textsc{SCFV}    & \textsc{SCFV+}    &  & \textsc{SCF}     & \textsc{SCFV}    & \textsc{SCFV+}    \\ \hline
web-Google                       &  & 24     & 33      & 32       &  & 48      & 67      & 66       \\
northwestUSA                     &  & 27     & 36      & 34       &  & 54      & 72      & 70       \\
wikitalk                         &  & 18     & 20      & 20       &  & 35      & 39      & 38       \\
greatlakes                       &  & 63     & 84      & 82       &  & 126     & 168     & 163      \\
fullUSA                          &  & 545    & 731     & 710      &  & 1085    & 1461    & 1417     \\ \hline
\end{tabular}\label{tb3}
\end{table}

\section{Conclusions}

In this paper, we addressed the problem of efficiently computing the diagonal of the forest matrix in digraphs. We proposed three novel sampling-based algorithms: \textsc{SCF}, \textsc{SCFV}, and \textsc{SCFV+}.  \textsc{SCF} utilizes an extension of Wilson's algorithm, based on a probabilistic interpretation of the diagonal of the forest matrix. \textsc{SCFV} is inspired by the FJ model, and refines \textsc{SCF} by reducing the variance in forest sampling through the matrix-vector iteration. \textsc{SCFV+} further reduces the variance using a novel iteration equation. Notably, \textsc{SCFV+} achieves a relative error guarantee with high probability and maintains a linear time complexity relative to the number of nodes, presenting a superior theoretical result compared to existing algorithms.

We conducted extensive experiments on various real-world networks. The results show that our algorithms demonstrate  superior effectiveness and efficiency compared to the state-of-the-art algorithms in undirected graphs. While state-of-the-art algorithms fail for digraphs,  our algorithms consistently perform well. Moreover, our algorithms are scalable to massive graphs with over thirty million nodes. In  future, we plan to extend or improve our algorithms to sign graphs or temporal graphs.
    
\begin{acks}
The work was supported by  the National Natural Science Foundation of China (Nos.  62372112 and U20B2051).
\end{acks}

\bibliographystyle{ACM-Reference-Format}
\balance
\bibliography{expressedopinion}


\providecommand{\noopsort}[1]{}\providecommand{\singleletter}[1]{#1}%\providecommand{\noopsort}[1]{}\providecommand{\singleletter}[1]{#1}%\providecommand{\noopsort}[1]{}\providecommand{\singleletter}[1]{#1}%\providecommand{\noopsort}[1]{}\providecommand{\singleletter}[1]{#1}%
\begin{thebibliography}{48}

%%% the Web address.

\ifx \showCODEN    \undefined \def \showCODEN     #1{\unskip}     \fi
\ifx \showDOI      \undefined \def \showDOI       #1{#1}\fi
\ifx \showISBNx    \undefined \def \showISBNx     #1{\unskip}     \fi
\ifx \showISBNxiii \undefined \def \showISBNxiii  #1{\unskip}     \fi
\ifx \showISSN     \undefined \def \showISSN      #1{\unskip}     \fi
\ifx \showLCCN     \undefined \def \showLCCN      #1{\unskip}     \fi
\ifx \shownote     \undefined \def \shownote      #1{#1}          \fi
\ifx \showarticletitle \undefined \def \showarticletitle #1{#1}   \fi
\ifx \showURL      \undefined \def \showURL       {\relax}        \fi
% invisible to TeX
\providecommand\bibfield[2]{#2}
\providecommand\bibinfo[2]{#2}
\providecommand\natexlab[1]{#1}
\providecommand\showeprint[2][]{arXiv:#2}

\bibitem[\protect\citeauthoryear{Achlioptas}{Achlioptas}{2003}]%
        {Ac03}
\bibfield{author}{\bibinfo{person}{Dimitris Achlioptas}.}
  \bibinfo{year}{2003}\natexlab{}.
\newblock \showarticletitle{{Database-friendly random projections:
  Johnson-Lindenstrauss with binary coins}}.
\newblock \bibinfo{journal}{\emph{J. Comput. System Sci.}}
  \bibinfo{volume}{66}, \bibinfo{number}{4} (\bibinfo{year}{2003}),
  \bibinfo{pages}{671--687}.
\newblock


\bibitem[\protect\citeauthoryear{Agaev and Chebotarev}{Agaev and
  Chebotarev}{2001}]%
        {AgCh01}
\bibfield{author}{\bibinfo{person}{Rafig~Pashaevich Agaev} {and}
  \bibinfo{person}{P~Yu Chebotarev}.} \bibinfo{year}{2001}\natexlab{}.
\newblock \showarticletitle{Spanning forests of a digraph and their
  applications}.
\newblock \bibinfo{journal}{\emph{Automation and Remote Control}}
  \bibinfo{volume}{62}, \bibinfo{number}{3} (\bibinfo{year}{2001}),
  \bibinfo{pages}{443--466}.
\newblock


\bibitem[\protect\citeauthoryear{Avena, Castell, Gaudilli{\`e}re, and
  M{\'e}lot}{Avena et~al\mbox{.}}{2018}]%
        {AvCaGaMe18}
\bibfield{author}{\bibinfo{person}{Luca Avena}, \bibinfo{person}{Fabienne
  Castell}, \bibinfo{person}{Alexandre Gaudilli{\`e}re}, {and}
  \bibinfo{person}{Clothilde M{\'e}lot}.} \bibinfo{year}{2018}\natexlab{}.
\newblock \showarticletitle{Random forests and networks analysis}.
\newblock \bibinfo{journal}{\emph{Journal of Statistical Physics}}
  \bibinfo{volume}{173} (\bibinfo{year}{2018}), \bibinfo{pages}{985--1027}.
\newblock


\bibitem[\protect\citeauthoryear{Avena and Gaudilli{\`e}re}{Avena and
  Gaudilli{\`e}re}{2018}]%
        {AvLuGaAl18}
\bibfield{author}{\bibinfo{person}{Luca Avena} {and} \bibinfo{person}{Alexandre
  Gaudilli{\`e}re}.} \bibinfo{year}{2018}\natexlab{}.
\newblock \showarticletitle{Two applications of random spanning forests}.
\newblock \bibinfo{journal}{\emph{Journal of Theoretical Probability}}
  \bibinfo{volume}{31}, \bibinfo{number}{4} (\bibinfo{year}{2018}),
  \bibinfo{pages}{1975--2004}.
\newblock


\bibitem[\protect\citeauthoryear{Bao and Zhang}{Bao and Zhang}{2022}]%
        {BaZh22}
\bibfield{author}{\bibinfo{person}{Qi Bao} {and} \bibinfo{person}{Zhongzhi
  Zhang}.} \bibinfo{year}{2022}\natexlab{}.
\newblock \showarticletitle{Discriminating power of centrality measures in
  complex networks}.
\newblock \bibinfo{journal}{\emph{IEEE Transactions on Cybernetics}}
  \bibinfo{volume}{52}, \bibinfo{number}{11} (\bibinfo{year}{2022}),
  \bibinfo{pages}{12583--12593}.
\newblock


\bibitem[\protect\citeauthoryear{Bavelas}{Bavelas}{1948}]%
        {Ba48}
\bibfield{author}{\bibinfo{person}{Alex Bavelas}.}
  \bibinfo{year}{1948}\natexlab{}.
\newblock \showarticletitle{A mathematical model for group structures}.
\newblock \bibinfo{journal}{\emph{Human organization}} \bibinfo{volume}{7},
  \bibinfo{number}{3} (\bibinfo{year}{1948}), \bibinfo{pages}{16--30}.
\newblock


\bibitem[\protect\citeauthoryear{Bavelas}{Bavelas}{1950}]%
        {Ba50}
\bibfield{author}{\bibinfo{person}{Alex Bavelas}.}
  \bibinfo{year}{1950}\natexlab{}.
\newblock \showarticletitle{Communication patterns in task-oriented groups}.
\newblock \bibinfo{journal}{\emph{The Journal of the Acoustical Society of
  America}} \bibinfo{volume}{22}, \bibinfo{number}{6} (\bibinfo{year}{1950}),
  \bibinfo{pages}{725--730}.
\newblock


\bibitem[\protect\citeauthoryear{Bindel, Kleinberg, and Oren}{Bindel
  et~al\mbox{.}}{2015}]%
        {BiKlOr15}
\bibfield{author}{\bibinfo{person}{David Bindel}, \bibinfo{person}{Jon
  Kleinberg}, {and} \bibinfo{person}{Sigal Oren}.}
  \bibinfo{year}{2015}\natexlab{}.
\newblock \showarticletitle{How bad is forming your own opinion?}
\newblock \bibinfo{journal}{\emph{Games and Economic Behavior}}
  \bibinfo{volume}{92} (\bibinfo{year}{2015}), \bibinfo{pages}{248--265}.
\newblock


\bibitem[\protect\citeauthoryear{Chaiken}{Chaiken}{1982}]%
        {Ch82}
\bibfield{author}{\bibinfo{person}{Seth Chaiken}.}
  \bibinfo{year}{1982}\natexlab{}.
\newblock \showarticletitle{A combinatorial proof of the all minors matrix tree
  theorem}.
\newblock \bibinfo{journal}{\emph{SIAM J. Alg. Disc. Meth.}}
  \bibinfo{volume}{3}, \bibinfo{number}{3} (\bibinfo{date}{Sep.}
  \bibinfo{year}{1982}), \bibinfo{pages}{319--329}.
\newblock


\bibitem[\protect\citeauthoryear{Chartrand, Schultz, and Winters}{Chartrand
  et~al\mbox{.}}{1996}]%
        {ChSchWe96}
\bibfield{author}{\bibinfo{person}{Gary Chartrand}, \bibinfo{person}{Michelle
  Schultz}, {and} \bibinfo{person}{Steven~J Winters}.}
  \bibinfo{year}{1996}\natexlab{}.
\newblock \showarticletitle{On eccentric vertices in graphs}.
\newblock \bibinfo{journal}{\emph{Networks}} \bibinfo{volume}{28},
  \bibinfo{number}{4} (\bibinfo{year}{1996}), \bibinfo{pages}{181--186}.
\newblock


\bibitem[\protect\citeauthoryear{Chebotarev and Agaev}{Chebotarev and
  Agaev}{2002}]%
        {ChAg02}
\bibfield{author}{\bibinfo{person}{Pavel Chebotarev} {and}
  \bibinfo{person}{Rafig Agaev}.} \bibinfo{year}{2002}\natexlab{}.
\newblock \showarticletitle{Forest matrices around the Laplacian matrix}.
\newblock \bibinfo{journal}{\emph{Linear Algebra Appl.}} \bibinfo{volume}{356},
  \bibinfo{number}{1-3} (\bibinfo{year}{2002}), \bibinfo{pages}{253--274}.
\newblock


\bibitem[\protect\citeauthoryear{Chebotarev and Shamis}{Chebotarev and
  Shamis}{2006}]%
        {ChSh06}
\bibfield{author}{\bibinfo{person}{Pavel~Yu. Chebotarev} {and}
  \bibinfo{person}{Elena Shamis}.} \bibinfo{year}{2006}\natexlab{}.
\newblock \showarticletitle{Matrix-Forest Theorems}.
\newblock \bibinfo{journal}{\emph{ArXiv}}  \bibinfo{volume}{abs/math/0602575}
  (\bibinfo{year}{2006}).
\newblock


\bibitem[\protect\citeauthoryear{Chebotarev and Shamis}{Chebotarev and
  Shamis}{1997}]%
        {ChSh97}
\bibfield{author}{\bibinfo{person}{P.~Yu Chebotarev} {and}
  \bibinfo{person}{E.~V. Shamis}.} \bibinfo{year}{1997}\natexlab{}.
\newblock \showarticletitle{The matrix-forest theorem and measuring relations
  in small social groups}.
\newblock \bibinfo{journal}{\emph{Automation and Remote Control}}
  \bibinfo{volume}{58}, \bibinfo{number}{9} (\bibinfo{year}{1997}),
  \bibinfo{pages}{1505--1514}.
\newblock


\bibitem[\protect\citeauthoryear{Chebotarev and Shamis}{Chebotarev and
  Shamis}{1998}]%
        {ChSh98}
\bibfield{author}{\bibinfo{person}{P.~Yu Chebotarev} {and}
  \bibinfo{person}{E.~V. Shamis}.} \bibinfo{year}{1998}\natexlab{}.
\newblock \showarticletitle{On proximity measures for graph vertices}.
\newblock \bibinfo{journal}{\emph{Automation and Remote Control}}
  \bibinfo{volume}{59}, \bibinfo{number}{10} (\bibinfo{year}{1998}),
  \bibinfo{pages}{1443--1459}.
\newblock


\bibitem[\protect\citeauthoryear{Chung and Lu}{Chung and Lu}{2006}]%
        {ChLu06}
\bibfield{author}{\bibinfo{person}{Fan Chung} {and} \bibinfo{person}{Linyuan
  Lu}.} \bibinfo{year}{2006}\natexlab{}.
\newblock \showarticletitle{Concentration inequalities and martingale
  inequalities: a survey}.
\newblock \bibinfo{journal}{\emph{Internet mathematics}} \bibinfo{volume}{3},
  \bibinfo{number}{1} (\bibinfo{year}{2006}), \bibinfo{pages}{79--127}.
\newblock


\bibitem[\protect\citeauthoryear{Cohen, Kyng, Miller, Pachocki, Peng, Rao, and
  Xu}{Cohen et~al\mbox{.}}{2014}]%
        {CoKyMiPaJaPeRaXu14}
\bibfield{author}{\bibinfo{person}{Michael~B Cohen}, \bibinfo{person}{Rasmus
  Kyng}, \bibinfo{person}{Gary~L Miller}, \bibinfo{person}{Jakub~W Pachocki},
  \bibinfo{person}{Richard Peng}, \bibinfo{person}{Anup~B Rao}, {and}
  \bibinfo{person}{Shen~Chen Xu}.} \bibinfo{year}{2014}\natexlab{}.
\newblock \showarticletitle{Solving {SDD} linear systems in nearly $m
  \log^{1/2} n$ time}. In \bibinfo{booktitle}{\emph{Proceedings of the
  Forty-Sixth Annual ACM Symposium on Theory of Computing}}. ACM,
  \bibinfo{pages}{343--352}.
\newblock


\bibitem[\protect\citeauthoryear{Daly and Haahr}{Daly and Haahr}{2007}]%
        {DaHa07}
\bibfield{author}{\bibinfo{person}{Elizabeth~M Daly} {and}
  \bibinfo{person}{Mads Haahr}.} \bibinfo{year}{2007}\natexlab{}.
\newblock \showarticletitle{Social network analysis for routing in disconnected
  delay-tolerant manets}. In \bibinfo{booktitle}{\emph{Proceedings of the 8th
  ACM International Symposium on Mobile Ad hoc Networking and Computing}}. ACM,
  \bibinfo{pages}{32--40}.
\newblock


\bibitem[\protect\citeauthoryear{Freeman}{Freeman}{1977}]%
        {Fr77}
\bibfield{author}{\bibinfo{person}{Linton~C Freeman}.}
  \bibinfo{year}{1977}\natexlab{}.
\newblock \showarticletitle{A set of measures of centrality based on
  betweenness}.
\newblock \bibinfo{journal}{\emph{Sociometry}} (\bibinfo{year}{1977}),
  \bibinfo{pages}{35--41}.
\newblock


\bibitem[\protect\citeauthoryear{Friedkin and Johnsen}{Friedkin and
  Johnsen}{1990}]%
        {FrJo90}
\bibfield{author}{\bibinfo{person}{Noah~E Friedkin} {and}
  \bibinfo{person}{Eugene~C Johnsen}.} \bibinfo{year}{1990}\natexlab{}.
\newblock \showarticletitle{Social influence and opinions}.
\newblock \bibinfo{journal}{\emph{Journal of Mathematical Sociology}}
  \bibinfo{volume}{15}, \bibinfo{number}{3-4} (\bibinfo{year}{1990}),
  \bibinfo{pages}{193--206}.
\newblock


\bibitem[\protect\citeauthoryear{Gionis, Terzi, and Tsaparas}{Gionis
  et~al\mbox{.}}{2013}]%
        {GiTeTs13}
\bibfield{author}{\bibinfo{person}{Aristides Gionis}, \bibinfo{person}{Evimaria
  Terzi}, {and} \bibinfo{person}{Panayiotis Tsaparas}.}
  \bibinfo{year}{2013}\natexlab{}.
\newblock \showarticletitle{Opinion maximization in social networks}. In
  \bibinfo{booktitle}{\emph{Proceedings of the 2013 SIAM International
  Conference on Data Mining}}. SIAM, \bibinfo{pages}{387--395}.
\newblock


\bibitem[\protect\citeauthoryear{Grando, Granville, and Lamb}{Grando
  et~al\mbox{.}}{2018}]%
        {GrGrLa18}
\bibfield{author}{\bibinfo{person}{Felipe Grando}, \bibinfo{person}{Lisandro~Z
  Granville}, {and} \bibinfo{person}{Luis~C Lamb}.}
  \bibinfo{year}{2018}\natexlab{}.
\newblock \showarticletitle{{Machine learning in network centrality measures:
  Tutorial and outlook}}.
\newblock \bibinfo{journal}{\emph{ACM Computing Surveys (CSUR)}}
  \bibinfo{volume}{51}, \bibinfo{number}{5} (\bibinfo{year}{2018}),
  \bibinfo{pages}{102}.
\newblock


\bibitem[\protect\citeauthoryear{Huang, Liao, and Wu}{Huang
  et~al\mbox{.}}{2018}]%
        {HuLiWu18}
\bibfield{author}{\bibinfo{person}{Lei Huang}, \bibinfo{person}{Li Liao}, {and}
  \bibinfo{person}{Cathy~H Wu}.} \bibinfo{year}{2018}\natexlab{}.
\newblock \showarticletitle{Completing sparse and disconnected protein-protein
  network by deep learning}.
\newblock \bibinfo{journal}{\emph{BMC Bioinformatics}} \bibinfo{volume}{19},
  \bibinfo{number}{1} (\bibinfo{year}{2018}), \bibinfo{pages}{103}.
\newblock


\bibitem[\protect\citeauthoryear{Jin, Bao, and Zhang}{Jin
  et~al\mbox{.}}{2019}]%
        {JiBaZh19}
\bibfield{author}{\bibinfo{person}{Yujia Jin}, \bibinfo{person}{Qi Bao}, {and}
  \bibinfo{person}{Zhongzhi Zhang}.} \bibinfo{year}{2019}\natexlab{}.
\newblock \showarticletitle{Forest distance closeness centrality in
  disconnected graphs}. In \bibinfo{booktitle}{\emph{2018 IEEE International
  Conference on Data Mining}}. IEEE, \bibinfo{pages}{339--348}.
\newblock


\bibitem[\protect\citeauthoryear{Johnson and Lindenstrauss}{Johnson and
  Lindenstrauss}{1984}]%
        {JoLi84}
\bibfield{author}{\bibinfo{person}{William~B Johnson} {and}
  \bibinfo{person}{Joram Lindenstrauss}.} \bibinfo{year}{1984}\natexlab{}.
\newblock \showarticletitle{{Extensions of Lipschitz mappings into a Hilbert
  space}}.
\newblock \bibinfo{journal}{\emph{Contemp. Math.}}  \bibinfo{volume}{26}
  (\bibinfo{year}{1984}), \bibinfo{pages}{189--206}.
\newblock


\bibitem[\protect\citeauthoryear{Kulesza, Taskar, et~al\mbox{.}}{Kulesza
  et~al\mbox{.}}{2012}]%
        {KuTa12}
\bibfield{author}{\bibinfo{person}{Alex Kulesza}, \bibinfo{person}{Ben Taskar},
  {et~al\mbox{.}}} \bibinfo{year}{2012}\natexlab{}.
\newblock \showarticletitle{Determinantal point processes for machine
  learning}.
\newblock \bibinfo{journal}{\emph{Foundations and Trends{\textregistered} in
  Machine Learning}} \bibinfo{volume}{5}, \bibinfo{number}{2--3}
  (\bibinfo{year}{2012}), \bibinfo{pages}{123--286}.
\newblock


\bibitem[\protect\citeauthoryear{Kunegis}{Kunegis}{2013}]%
        {Ku13}
\bibfield{author}{\bibinfo{person}{J{\'e}r{\^o}me Kunegis}.}
  \bibinfo{year}{2013}\natexlab{}.
\newblock \showarticletitle{Konect: the koblenz network collection}. In
  \bibinfo{booktitle}{\emph{Proceedings of the 22nd International World Wide
  Web Conference}}. ACM, \bibinfo{pages}{1343--1350}.
\newblock


\bibitem[\protect\citeauthoryear{Lawler and Gregory}{Lawler and
  Gregory}{1980}]%
        {La80}
\bibfield{author}{\bibinfo{person}{Lawler} {and} \bibinfo{person}{F. Gregory}.}
  \bibinfo{year}{1980}\natexlab{}.
\newblock \showarticletitle{A self-avoiding random walk}.
\newblock \bibinfo{journal}{\emph{Duke Mathematical Journal}}
  \bibinfo{volume}{47}, \bibinfo{number}{3} (\bibinfo{year}{1980}),
  \bibinfo{pages}{655--693}.
\newblock


\bibitem[\protect\citeauthoryear{Lawler}{Lawler}{1979}]%
        {LaFr79}
\bibfield{author}{\bibinfo{person}{Gregory~Francis Lawler}.}
  \bibinfo{year}{1979}\natexlab{}.
\newblock \emph{\bibinfo{title}{A self-avoiding random walk.}}
\newblock \bibinfo{thesistype}{Ph.D. Dissertation}. \bibinfo{school}{Princeton
  University}.
\newblock


\bibitem[\protect\citeauthoryear{Leskovec and Sosi{\v{c}}}{Leskovec and
  Sosi{\v{c}}}{2016}]%
        {LeSo16}
\bibfield{author}{\bibinfo{person}{Jure Leskovec} {and} \bibinfo{person}{Rok
  Sosi{\v{c}}}.} \bibinfo{year}{2016}\natexlab{}.
\newblock \showarticletitle{{SNAP: A} general-purpose network analysis and
  graph-mining library}.
\newblock \bibinfo{journal}{\emph{ACM Transactions on Intelligent Systems and
  Technology}} \bibinfo{volume}{8}, \bibinfo{number}{1} (\bibinfo{year}{2016}),
  \bibinfo{pages}{1}.
\newblock


\bibitem[\protect\citeauthoryear{Li, Peng, Shan, Yi, and Zhang}{Li
  et~al\mbox{.}}{2019}]%
        {LiPeShYiZh19}
\bibfield{author}{\bibinfo{person}{Huan Li}, \bibinfo{person}{Richard Peng},
  \bibinfo{person}{Liren Shan}, \bibinfo{person}{Yuhao Yi}, {and}
  \bibinfo{person}{Zhongzhi Zhang}.} \bibinfo{year}{2019}\natexlab{}.
\newblock \showarticletitle{Current flow group closeness centrality for complex
  networks}. In \bibinfo{booktitle}{\emph{Proceedings of World Wide Web
  Conference}}. ACM, \bibinfo{pages}{961--971}.
\newblock


\bibitem[\protect\citeauthoryear{Liao, Li, Dai, Chen, Qin, and Wang}{Liao
  et~al\mbox{.}}{2023a}]%
        {LiLiDaChQiWa23PageRank}
\bibfield{author}{\bibinfo{person}{Meihao Liao}, \bibinfo{person}{Rong-Hua Li},
  \bibinfo{person}{Qiangqiang Dai}, \bibinfo{person}{Hongyang Chen},
  \bibinfo{person}{Hongchao Qin}, {and} \bibinfo{person}{Guoren Wang}.}
  \bibinfo{year}{2023}\natexlab{a}.
\newblock \showarticletitle{Efficient Personalized PageRank Computation: The
  Power of Variance-Reduced Monte Carlo Approaches}.
\newblock \bibinfo{journal}{\emph{Proceedings of the ACM on Management of
  Data}} \bibinfo{volume}{1}, \bibinfo{number}{2} (\bibinfo{year}{2023}),
  \bibinfo{pages}{1--26}.
\newblock


\bibitem[\protect\citeauthoryear{Liao, Li, Dai, Chen, Qin, and Wang}{Liao
  et~al\mbox{.}}{2023b}]%
        {LiLiDaChQiWa23Resistance}
\bibfield{author}{\bibinfo{person}{Meihao Liao}, \bibinfo{person}{Rong-Hua Li},
  \bibinfo{person}{Qiangqiang Dai}, \bibinfo{person}{Hongyang Chen},
  \bibinfo{person}{Hongchao Qin}, {and} \bibinfo{person}{Guoren Wang}.}
  \bibinfo{year}{2023}\natexlab{b}.
\newblock \showarticletitle{Efficient Resistance Distance Computation: the
  Power of Landmark-based Approaches}.
\newblock \bibinfo{journal}{\emph{Proceedings of the ACM on Management of
  Data}} \bibinfo{volume}{1}, \bibinfo{number}{1} (\bibinfo{year}{2023}),
  \bibinfo{pages}{1--27}.
\newblock


\bibitem[\protect\citeauthoryear{Liao, Li, Dai, and Wang}{Liao
  et~al\mbox{.}}{2022}]%
        {LiLiDaWa22}
\bibfield{author}{\bibinfo{person}{Meihao Liao}, \bibinfo{person}{Rong-Hua Li},
  \bibinfo{person}{Qiangqiang Dai}, {and} \bibinfo{person}{Guoren Wang}.}
  \bibinfo{year}{2022}\natexlab{}.
\newblock \showarticletitle{Efficient Personalized PageRank Computation: A
  Spanning Forests Sampling Based Approach}. In
  \bibinfo{booktitle}{\emph{Proceedings of the 2022 International Conference on
  Management of Data}} (Philadelphia, PA, USA)
  \emph{(\bibinfo{series}{SIGMOD/PODS '22})}. \bibinfo{publisher}{Association
  for Computing Machinery}, \bibinfo{address}{New York, NY, USA},
  \bibinfo{pages}{2048–2061}.
\newblock


\bibitem[\protect\citeauthoryear{Marchal}{Marchal}{2000}]%
        {Ma00}
\bibfield{author}{\bibinfo{person}{Philippe Marchal}.}
  \bibinfo{year}{2000}\natexlab{}.
\newblock \showarticletitle{Loop-erased random walks, spanning trees and
  Hamiltonian cycles}.
\newblock \bibinfo{journal}{\emph{Electronic Communications in Probability}}
  \bibinfo{volume}{5} (\bibinfo{year}{2000}), \bibinfo{pages}{39--50}.
\newblock


\bibitem[\protect\citeauthoryear{Merris}{Merris}{1994}]%
        {Me94}
\bibfield{author}{\bibinfo{person}{Russell Merris}.}
  \bibinfo{year}{1994}\natexlab{}.
\newblock \showarticletitle{Laplacian matrices of graphs: a survey}.
\newblock \bibinfo{journal}{\emph{Linear Algebra Appl.}}  \bibinfo{volume}{197}
  (\bibinfo{year}{1994}), \bibinfo{pages}{143--176}.
\newblock


\bibitem[\protect\citeauthoryear{Newman}{Newman}{2010}]%
        {Ne10}
\bibfield{author}{\bibinfo{person}{Mark E.~J. Newman}.}
  \bibinfo{year}{2010}\natexlab{}.
\newblock \bibinfo{booktitle}{\emph{Networks: An Introduction}}.
\newblock \bibinfo{publisher}{Oxford University Press, Oxford, UK}.
\newblock


\bibitem[\protect\citeauthoryear{Pilavci, Amblard, Barthelme, and
  Tremblay}{Pilavci et~al\mbox{.}}{2020}]%
        {PiAmBaTr20}
\bibfield{author}{\bibinfo{person}{Yusuf~Y Pilavci},
  \bibinfo{person}{Pierre-Olivier Amblard}, \bibinfo{person}{Simon Barthelme},
  {and} \bibinfo{person}{Nicolas Tremblay}.} \bibinfo{year}{2020}\natexlab{}.
\newblock \showarticletitle{Smoothing graph signals via random spanning
  forests}. In \bibinfo{booktitle}{\emph{IEEE International Conference on
  Acoustics, Speech and Signal Processing}}. IEEE, \bibinfo{pages}{5630--5634}.
\newblock


\bibitem[\protect\citeauthoryear{Pilavc{\i}, Amblard, Barthelme, and
  Tremblay}{Pilavc{\i} et~al\mbox{.}}{2021}]%
        {PiAmBaTr21}
\bibfield{author}{\bibinfo{person}{Yusuf~Yi{\u{g}}it Pilavc{\i}},
  \bibinfo{person}{Pierre-Olivier Amblard}, \bibinfo{person}{Simon Barthelme},
  {and} \bibinfo{person}{Nicolas Tremblay}.} \bibinfo{year}{2021}\natexlab{}.
\newblock \showarticletitle{Graph {T}ikhonov regularization and interpolation
  via random spanning forests}.
\newblock \bibinfo{journal}{\emph{IEEE transactions on Signal and Information
  Processing over Networks}}  \bibinfo{volume}{7} (\bibinfo{year}{2021}),
  \bibinfo{pages}{359--374}.
\newblock


\bibitem[\protect\citeauthoryear{Pilavci, Amblard, Barthelme, and
  Tremblay}{Pilavci et~al\mbox{.}}{2022}]%
        {PiAmBaTr22trace}
\bibfield{author}{\bibinfo{person}{Yusuf~Yigit Pilavci},
  \bibinfo{person}{Pierre-Olivier Amblard}, \bibinfo{person}{Simon Barthelme},
  {and} \bibinfo{person}{Nicolas Tremblay}.} \bibinfo{year}{2022}\natexlab{}.
\newblock \showarticletitle{Variance reduction for inverse trace estimation via
  random spanning forests}.
\newblock \bibinfo{journal}{\emph{arXiv preprint arXiv:2206.07421}}
  (\bibinfo{year}{2022}).
\newblock


\bibitem[\protect\citeauthoryear{Pilavc{\i}, Amblard, Barthelm{\'e}, and
  Tremblay}{Pilavc{\i} et~al\mbox{.}}{2022}]%
        {PiAmBaTr22}
\bibfield{author}{\bibinfo{person}{Yusuf~Yigit Pilavc{\i}},
  \bibinfo{person}{Pierre-Olivier Amblard}, \bibinfo{person}{Simon
  Barthelm{\'e}}, {and} \bibinfo{person}{Nicolas Tremblay}.}
  \bibinfo{year}{2022}\natexlab{}.
\newblock \showarticletitle{Variance Reduction in Stochastic Methods for
  Large-Scale Regularized Least-Squares Problems}. In
  \bibinfo{booktitle}{\emph{2022 30th European Signal Processing Conference
  (EUSIPCO)}}. IEEE, \bibinfo{pages}{1771--1775}.
\newblock


\bibitem[\protect\citeauthoryear{Rossi, Frasca, and Fagnani}{Rossi
  et~al\mbox{.}}{2017}]%
        {RoFrFa17}
\bibfield{author}{\bibinfo{person}{Wilbert~Samuel Rossi},
  \bibinfo{person}{Paolo Frasca}, {and} \bibinfo{person}{Fabio Fagnani}.}
  \bibinfo{year}{2017}\natexlab{}.
\newblock \showarticletitle{Distributed estimation from relative and absolute
  measurements}.
\newblock \bibinfo{journal}{\emph{IEEE Trans. Automat. Control}}
  \bibinfo{volume}{62}, \bibinfo{number}{12} (\bibinfo{year}{2017}),
  \bibinfo{pages}{6385--6391}.
\newblock


\bibitem[\protect\citeauthoryear{Saad and Schultz}{Saad and Schultz}{1986}]%
        {SaSc86}
\bibfield{author}{\bibinfo{person}{Youcef Saad} {and} \bibinfo{person}{Martin~H
  Schultz}.} \bibinfo{year}{1986}\natexlab{}.
\newblock \showarticletitle{GMRES: A generalized minimal residual algorithm for
  solving nonsymmetric linear systems}.
\newblock \bibinfo{journal}{\emph{SIAM Journal on scientific and statistical
  computing}} \bibinfo{volume}{7}, \bibinfo{number}{3} (\bibinfo{year}{1986}),
  \bibinfo{pages}{856--869}.
\newblock


\bibitem[\protect\citeauthoryear{Skibski, Michalak, and Rahwan}{Skibski
  et~al\mbox{.}}{2018}]%
        {SkMiRa18}
\bibfield{author}{\bibinfo{person}{Oskar Skibski}, \bibinfo{person}{Tomasz~P
  Michalak}, {and} \bibinfo{person}{Talal Rahwan}.}
  \bibinfo{year}{2018}\natexlab{}.
\newblock \showarticletitle{Axiomatic Characterization of Game-Theoretic
  Centrality}.
\newblock \bibinfo{journal}{\emph{Journal of Artificial Intelligence Research}}
   \bibinfo{volume}{62} (\bibinfo{year}{2018}), \bibinfo{pages}{33--68}.
\newblock


\bibitem[\protect\citeauthoryear{Skibski, Rahwan, Michalak, and Yokoo}{Skibski
  et~al\mbox{.}}{2019}]%
        {SkRaMiYo19}
\bibfield{author}{\bibinfo{person}{Oskar Skibski}, \bibinfo{person}{Talal
  Rahwan}, \bibinfo{person}{Tomasz~P Michalak}, {and} \bibinfo{person}{Makoto
  Yokoo}.} \bibinfo{year}{2019}\natexlab{}.
\newblock \showarticletitle{Attachment Centrality: {M}easure for Connectivity
  in Networks}.
\newblock \bibinfo{journal}{\emph{Artificial Intelligence}}
  \bibinfo{volume}{274} (\bibinfo{year}{2019}), \bibinfo{pages}{151--179}.
\newblock


\bibitem[\protect\citeauthoryear{Sun and Zhang}{Sun and Zhang}{2023}]%
        {SuZh23}
\bibfield{author}{\bibinfo{person}{Haoxin Sun} {and} \bibinfo{person}{Zhongzhi
  Zhang}.} \bibinfo{year}{2023}\natexlab{}.
\newblock \showarticletitle{Opinion Optimization in Directed Social Networks}.
  In \bibinfo{booktitle}{\emph{Proceedings of the AAAI Conference on Artificial
  Intelligence}}, Vol.~\bibinfo{volume}{37}. \bibinfo{pages}{4623--4632}.
\newblock


\bibitem[\protect\citeauthoryear{van~der Grinten, Angriman, Predari, and
  Meyerhenke}{van~der Grinten et~al\mbox{.}}{2021}]%
        {GrAnPrMe21}
\bibfield{author}{\bibinfo{person}{Alexander van~der Grinten},
  \bibinfo{person}{Eugenio Angriman}, \bibinfo{person}{Maria Predari}, {and}
  \bibinfo{person}{Henning Meyerhenke}.} \bibinfo{year}{2021}\natexlab{}.
\newblock \showarticletitle{New Approximation Algorithms for Forest Closeness
  Centrality--for Individual Vertices and Vertex Groups}. In
  \bibinfo{booktitle}{\emph{Proceedings of the 2021 SIAM International
  Conference on Data Mining}}. SIAM, \bibinfo{pages}{136--144}.
\newblock


\bibitem[\protect\citeauthoryear{Wilson}{Wilson}{1996}]%
        {Wi96}
\bibfield{author}{\bibinfo{person}{David~Bruce Wilson}.}
  \bibinfo{year}{1996}\natexlab{}.
\newblock \showarticletitle{Generating random spanning trees more quickly than
  the cover time}. In \bibinfo{booktitle}{\emph{Proceedings of the
  Twenty-Eighth Annual ACM Symposium on Theory of Computing}}.
  \bibinfo{pages}{296–303}.
\newblock


\bibitem[\protect\citeauthoryear{Xu, Bao, and Zhang}{Xu et~al\mbox{.}}{2021}]%
        {XuBaZh21}
\bibfield{author}{\bibinfo{person}{Wanyue Xu}, \bibinfo{person}{Qi Bao}, {and}
  \bibinfo{person}{Zhongzhi Zhang}.} \bibinfo{year}{2021}\natexlab{}.
\newblock \showarticletitle{Fast evaluation for relevant quantities of opinion
  dynamics}. In \bibinfo{booktitle}{\emph{Proceedings of The Web Conference}}.
  ACM, \bibinfo{pages}{2037--2045}.
\newblock


\end{thebibliography}

\clearpage
\appendix
\section{APPENDIX}
In this section we provide  proofs of  lemmas and theorems in the article. 
\subsection{Chernoff Bound}
\begin{lemma}(Chernoff bound)\label{le-chernoff}
	Let $ x_i(1\leq i\leq l) $ be independent random variables satisfying $ |x_i- \mathbb{E}\{x_i\}|\leq M$ for all $ 1\leq i \leq l $. Let $ x = \frac{1}{l} \sum_{i=1}^l x_i $. Then we have
	\begin{equation}\label{key}
		\mathbb{P}\{|x-\mathbb{E}\{x\}| \leq \epsilon \}\geq 1-2\exp{\left(-\frac{l\epsilon^2}{2({\rm Var}\{x\}l+M\epsilon/3)}\right)}.
	\end{equation}
\end{lemma}

\subsection{Proof of Theorem~\ref{th-reciprocal}}
\begin{proof}
Given that $\sum_{j=1}^n\omega_{ji}=1$ and $\omega_{ij}= |\calF_{ij}|/|\calF|$, the following can be deduced:
\begin{align}\label{rhonew}
\frac{1}{\omega_{ii}} &= \frac{\sum_{j=1}^n\omega_{ji}}{\omega_{ii}}  = \frac{1}{|\calF_{ii}|}\sum_{j=1}^n  |\calF_{ji}|
= \frac{1}{|\calF_{ii}|}\sum\limits_{j=1}^{n}\sum\limits_{\phi\in \calF_{ii}}\mathbb{I}_{\{r_{\phi}(j)=i\}}\notag\\ 
& = \frac{1}{|\calF_{ii}|}\sum\limits_{\phi\in \calF_{ii}}\sum\limits_{j=1}^{n}\mathbb{I}_{\{r_{\phi}(j)=i\}}
= \frac{1}{|\calF_{ii}|}\sum\limits_{\phi\in \calF_{ii}}|N(\phi,i)|,
\end{align}
% \end{equation}

This implies that $\frac{1}{\omega_{ii}}$ represents the average number of nodes within the converging trees that are rooted at node $i$ across all $\phi\in \calF_{ii}$.
\end{proof}

\subsection{Proof of Lemma~\ref{le-rf}}
 \begin{proof}
 	Wilson showed that the expected running time of generating a uniform spanning tree of a connected digraph $\calG$ rooted at node $u$ is equal to a weighted average of commute times between the root and the other nodes~\cite{Wi96}. Marchal rewrote this average of commute times in terms of graph matrices in Proposition 1 in~\cite{Ma00}. In~\cite{SuZh23}, the author further analyzed the expected time complexity for the extension of Wilson's algorithm in  $ \calG'$ is $O(n)$. Thus, the expected time complexity of Algorithm 1 is $O(ln)$.
 \end{proof}

\subsection{Proof of Theorem~\ref{th-rf}}

\begin{proof}
    Recall that for a given spanning converging forest  $\phi \in \calF$  and a node $i \in V$, the random variable $\widehat{\omega}_{ii}(\phi)$ is defined as $\mathbb{I}_{\{i \in \calR(\phi)\}}$. Given that the possible values of $\widehat{\omega}_{ii}(\phi)$ are either 0 or 1, it follows that $|\widehat{\omega}_{ii}(\phi) - \omega_{ii}| \leq 1$.

    In   Algorithm~\ref{alg-rf}, we generate a set of $l$ spanning forests, represented as $\phi_1, \cdots, \phi_l$. The output $\widehat{\boldsymbol{\omega}}[i]$ of Algorithm~\ref{alg-rf} is  $\widehat{\boldsymbol{\omega}}[i] = \frac{1}{l} \sum_{j=1}^l \widehat{\omega}_{ii}(\phi_j)$.

Then, we can compute the variance of $\widehat{\boldsymbol{\omega}}[i]$ as:
\[ {\rm Var}\{\frac{1}{l} \sum_{j=1}^l \widehat{\omega}_{ii}(\phi_j)\} = \frac{1}{l} {\rm Var}\{\widehat{\omega}_{ii}(\phi)\} = \frac{1}{l}(\omega_{ii} - \omega_{ii}^2). \]

To obtain a relative error bound, it is necessary to satisfy the inequality:
 $$\mathbb{P}\{|\widehat{\boldsymbol{\omega}}[i]-\omega_{ii}| \geq \epsilon \omega_{ii} \}\leq \delta.$$

By invoking the Chernoff bound in Lemma~\ref{le-chernoff}, and designating $x_j = \widehat{\omega}_{ii}(\phi_j)$ for $1 \leq j \leq l$ and $x = \widehat{\boldsymbol{\omega}}[i]$, in order to prove~\eqref{eq-omegaii} we only need to show that the following inequality holds:

$$2\exp{\left(-\frac{l\epsilon^2\omega_{ii}^2}{2({\rm Var}\{\widehat{\omega}_{ii}\}+M\epsilon\omega_{ii}/3)}\right)} \leq \delta,$$
which leads to:
\begin{equation}\label{ltomeet}
l \geq \log(\frac{2}{\delta})\left(\frac{2{\rm Var}\{\widehat{\omega}_{ii}\}}{\epsilon^2\omega_{ii}^2}+\frac{2M}{3\epsilon\omega_{ii}}\right).
\end{equation}
Since $|\widehat{\omega}_{ii} - \omega_{ii}| \leq 1$, we can set $M=1$. Considering   ${\rm Var}\{ \widehat{\omega}_{ii}(\phi_j)\} = \omega_{ii} - \omega_{ii}^2$, the inequality~\eqref{ltomeet} reduces to:
$$l \geq \log(\frac{2}{\delta})\left(\frac{2}{\epsilon^2\omega_{ii}}
+\frac{2}{3\epsilon\omega_{ii}}-\frac{2}{\epsilon^2}\right).
$$
Thus, selecting $ l =\left \lceil \frac{6+2\epsilon}{3\sigma\epsilon^2}\log{\frac{2}{\delta} }  \right \rceil$ ensures that the required inequality always holds. This completes the proof.
\end{proof}

\subsection{Proof of Lemma~\ref{le-var-omega}}

\begin{proof}
    According to Lemma~\ref{le-rf}, $\boldsymbol{\rho}^i$ satisfying $\boldsymbol{\rho}^i = (\II+\DD)^{-1}\ee_i + (\II+\DD)^{-1}\AA\boldsymbol{\rho}_i$. Thus, $\ee_i^\top\boldsymbol{\rho}^i = \ee_i^\top(\II+\DD)^{-1}\ee_i + \ee_i^\top(\II+\DD)^{-1}\AA\boldsymbol{\rho}^i$, implying ${\omega}_{ii} =  \frac{1}{1+d_i}(1 + \sum_{j\in N(i)}{\omega}_{ji})$. Since $\widetilde{\omega}_{ii}(\phi) 
    = \frac{1}{1+d_i}(1 + \sum_{j\in N(i)}\widehat{\omega}_{ji}(\phi))$, and $\widehat{\omega}_{ji}(\phi)$ is an unbiased estimator of $\omega_{ji}$, $\widetilde{\omega}_{ii}$ is an unbiased estimator of $\omega_{ii}$.

 Next we  derive the variance of  estimator $\widetilde{\omega}_{ii}$. For a spanning converging forest $\phi \in \calF$ and two nodes $i,j\in V$ with $i\neq j$, we have $\widehat{\omega}_{ij}(\phi) = \mathbb{I}_{\{r_{\phi}(i)=j\}} = 1$ or $0$, leading to $\mathbb{E}\{\widehat{\omega}_{ij}^2\} = \mathbb{E}\{\widehat{\omega}_{ij}\} = {\omega}_{ij}$. Moreover, for distinct nodes $i,j,k\in V$, we have $\mathbb{E}\{\widehat{\omega}_{ji}(\phi)\widehat{\omega}_{ki}(\phi)\} = q_{jk}^{(i)}$. Thus, we can derive  
\begin{equation}
    \begin{aligned}
    &\quad {\rm Var}\{\widetilde{\omega}_{ii}(\phi)\}= \mathbb{E}\{(\widetilde{\omega}_{ii})^2\} - (\mathbb{E}\{\widetilde{\omega}_{ii}\})^2 \\&= \frac{1}{(1+d_i)^2}\mathbb{E}\{( 1+\sum_{j\in N(i)}\widehat{\omega}_{ji})^2 \} - \omega_{ii}^2 
     \\ &= \frac{1}{(1+d_i)^2}\mathbb{E}\{1+2\sum_{i\in N(j)}\widehat{\omega}_{ji} + (\sum_{i\in N(j)}\widehat{\omega}_{ji})^2 \} - \omega_{ii}^2 
     \\ &= \frac{1+3\sum_{i\in N(j)}{\omega}_{ji}}{(1+d_i)^2}+ \frac{2\sum_{j,k\in N(i)}q^{(i)}_{jk}}{(1+d_i)^2}-\omega_{ii}^2 \\&=  \frac{1+3((1+d_i)\omega_{ii}-1)}{(1+d_i)^2}+ \frac{2\sum_{j,k\in N(i)}q^{(i)}_{jk}}{(1+d_i)^2}- \omega_{ii}^2 \\ &=\frac{3\omega_{ii}}{1+d_i} - \frac{2}{(1+d_i)^2} + \frac{2\sum_{j,k\in N(i)}q^{(i)}_{jk}}{(1+d_i)^2}-\omega_{ii}^2.
    \end{aligned}
\end{equation}

For a spanning converging forest $\phi\in \calF$, and nodes $i,j\in V$ with $i\neq j$, we consider the scenario where both  $\widehat{\omega}_{ji}(\phi)$ and $\widehat{\omega}_{ki}(\phi)$  are equal to 1. This implies that node $i$  acts as a root node in  $\phi$, leading to  $\widehat{\omega}_{ii} = 1$ . From this observation, we can infer that the expected value of the product $ \widehat{\omega}_{ji}(\phi)\widehat{\omega}_{ki}(\phi) $ is bounded by the expected value of $ \widehat{\omega}_{ii}(\phi) $. Formally, this can be expressed as $ q_{jk}^{(i)} \leq \omega_{ii} $. Building upon this foundation, we can further extend our analysis to compare the variances of the two estimators $ \widehat{\omega}_{ii} $ and $ \widetilde{\omega}_{ii}^{(0)} $. The mathematical derivation is as follows:

\begin{equation}
    \begin{aligned}
        {\rm Var}\{\widehat{\omega}_{ii}\} - {\rm Var}\{\widetilde{\omega}_{ii}\} & = \omega_{ii} - \frac{3\omega_{ii}}{1+d_i} + \frac{2}{(1+d_i)^2} - \frac{2\sum_{j,k\in N(i)}q^{(i)}_{jk}}{(1+d_i)^2}\\ &\geq \omega_{ii} - \frac{3\omega_{ii}}{1+d_i} + \frac{2}{(1+d_i)^2} - \frac{d_{i}(d_i-1)\omega_{ii}}{(1+d_i)^2} \\ &  =\frac{2(1-\omega_{ii})}{(1+d_i)^2}\geq 0.
    \end{aligned}
\end{equation}
This derivation solidifies our understanding and confirms the reduced variance of the estimator  $\widetilde{\omega}_{ii}$  compared to $\widehat{\omega}_{ii}$, thus completing our proof.
\end{proof}
\subsection{Proof of Lemma~\ref{le-var-omega-new}}

\begin{proof}
Since $\boldsymbol{\gamma}^i$  is the equilibrium vector of the iteration equation~\eqref{FJ2}, it obeys relation ${\boldsymbol{\gamma}^i}^\top = \ee_i^\top(\II+\DD)^{-1} +{\boldsymbol{\gamma}^i}^\top\AA(\II+\DD)^{-1}$. Thus, ${\boldsymbol{\gamma}^i}^\top\ee_i = \ee_i^\top(\II+\DD)^{-1}\ee_i +{\boldsymbol{\gamma}^i}^\top\AA(\II+\DD)^{-1}\ee_i$, meaning ${\omega}_{ii} =  \frac{1}{1+d_i}(1 + \sum_{j: i\in N(j)}{\omega}_{ij})$. Since $\widebar{\omega}_{ii}(\phi) = \frac{1}{1+d_i} + \frac{1}{1+d_i}\sum_{j:i\in N(j)}\widehat{\omega}_{ij}(\phi) $ and $\widehat{\omega}_{ij}(\phi)$ is an unbiased estimator of $\omega_{ij}$, $\widebar{\omega}_{ii}$ is an unbiased estimator of $\omega_{ii}$.

% Starting with Equation \eqref{Omega0new}, it's straightforward to deduce that $\mathbb{E}\{  \widebar{\omega}_{ii}(\phi)\} = \omega_{ii}$. This establishes that $\widebar{\omega}_{ii}(\phi)$ acts as an unbiased estimator for $\omega_{ii}$.

To determine the variance of the random variable $\widebar{\omega}_{ii}(\phi)$, we follow a similar approach as used for $\widetilde{\omega}_{ii}(\phi)$ in Lemma~\ref{le-var-omega}. It's crucial to note that for distinct nodes $i,j,k\in V$, the relationship $\widehat{\omega}_{ij}(\phi)\widehat{\omega}_{ik}(\phi) = 0$ consistently holds for any spanning converging forest $\phi\in\calF$. This ensures that the cross-product term is eliminated from our calculations. Given that the cross-product term is non-negative, it's evident that $\widebar{\omega}(\phi)$ exhibits a reduced variance in comparison to the previous two random variables.
\end{proof}

\subsection{Proof of Lemma~\ref{le-varbound}}
\begin{proof}
    According to Lemma~\ref{le-var-omega-new}, we derive 
    \begin{equation}
    \begin{aligned}
              &\left|\frac{{\rm Var}\{\widebar{\omega}_{ii}\}}{\omega_{ii}^2}\right| = \frac{3}{(1+d_i)\omega_{ii}} - \frac{2}{(1+d_i)^2\omega_{ii}^2} -1 \\&= -\frac{2}{(1+d_i)^2}(\frac{1}{\omega_{ii}}-\frac{3(1+d_i)}{4})^2+\frac{1}{8}  \leq \frac{1}{8}
    \end{aligned}
    \end{equation} 
The equation holds when $\frac{1}{\omega_{ii}}=\frac{3(1+d_i)}{4}$,  implying $\omega_{ii} = \frac{4}{3(1+d_i)}$, which finishes the proof.
\end{proof}
\subsection{Proof of Theorem~\ref{th-Fast}}

\begin{proof}
For a given spanning converging forest  $\phi \in \calF$ and a node $i \in V$,  the random variable $\widebar{\omega}_{ii}(\phi)$ is defined as $      \widebar{\omega}_{ii}(\phi) = \frac{1}{1+d_i} + \frac{1}{1+d_i}\sum_{j:i\in N(j)}\widehat{\omega}_{ij}(\phi)$. Given that node $i$ has only one root node in $\phi$, the possible values of $\widebar{\omega}_{ii}(\phi)$ are either $\frac{1}{1+d_i}$ or $\frac{2}{1+d_i}$. Since $\frac{1}{1+d_i}\leq \omega_{ii} \leq \frac{2}{2+d_i}$, it follows that $|\widebar{\omega}_{ii} - \omega_{ii}| \leq \frac{1}{1+d_i}$.

In  Algorithm~\ref{alg-rfwithvar+}, we generate $l$ spanning forests  $\phi_1, \phi_2,\cdots, \phi_l$. Its output $\widebar{\boldsymbol{\omega}}[i]$ can expressed as:
 $\widebar{\boldsymbol{\omega}}[i] = \frac{1}{l} \sum_{j=1}^l \widebar{\omega}_{ii}(\phi_j)$.

% The variance of  $\widebar{\boldsymbol{\omega}}[i]$ can be calculated as:
% $$ {\rm Var}\{\frac{1}{l} \sum_{j=1}^l \widebar{\omega}_{ii}(\phi_j)\} = \frac{1}{l} {\rm Var}\{\widehat{\omega}_{ii}(\phi)\} = \frac{1}{l}(\frac{3\omega_{ii}}{1+d_i} - \frac{2}{(1+d_i)^2} -\omega_{ii}^2). $$

%  To obtain a relative error bound, it is necessary to satisfy the inequality:
%  $$\mathbb{P}\{|\widebar{\boldsymbol{\omega}}[i]-\omega_i| \geq \epsilon \omega_i \}\leq \delta.$$ 

%Now we can apply the Chernoff bound  in Lemma~\ref{le-chernoff}. 
Setting $x_j = \widebar{\omega}_{ii}(\phi_j)$  for $1 \leq j \leq l$,  $M = \frac{1}{1+d_i}$, and  $x = \widebar{\boldsymbol{\omega}}[i]$, and using a similar analysis as in the proof of Theorem~\ref{th-rf}, the number of $l$ needs to satisfy the following inequality:

% , we derive:
%  $$2\exp{\left(-\frac{l\epsilon^2\omega_{ii}^2}{2({\rm Var}\{\widebar{\omega}_{ii}\}+M\epsilon\omega_{ii}/3)}\right)} \leq \delta,$$ 

 \begin{equation}\label{l}
      l \geq \log(\frac{2}{\delta})\left(\frac{2{\rm Var}\{\widebar{\omega}_{ii}\}}{\epsilon^2\omega_{ii}^2}+\frac{2}{3(1+d_i)\epsilon\omega_{ii}}\right). 
 \end{equation}
Making  use of Lemma~\ref{le-varbound} we have $\frac{{\rm Var}\{\widebar{\omega}_{ii}\}}{\omega_{ii}^2}\leq \frac{1}{8}$ and $ (1+d_i)\omega_{ii}\leq 1$.
Thus, selecting $l =\left \lceil  (\frac{2}{3\epsilon}+\frac{1}{4\epsilon^2})\log(\frac{2}{\delta})  \right \rceil$ ensures that inequality~\eqref{l} always holds. This completes the proof.
\end{proof}

\end{document}